\documentclass[acmsmall]{acmart}

\def\BibTeX{{\rm B\kern-.05em{\sc i\kern-.025em b}\kern-.08emT\kern-.1667em\lower.7ex\hbox{E}\kern-.125emX}}

\copyrightyear{2018}
\acmYear{2018}
\setcopyright{acmlicensed}
\acmConference[Woodstock '18]{Woodstock '18: ACM Symposium on Neural Gaze Detection}{June 03--05, 2018}{Woodstock, NY}
\acmBooktitle{Woodstock '18: ACM Symposium on Neural Gaze Detection, June 03--05, 2018, Woodstock, NY}
\acmPrice{15.00}
\acmDOI{10.1145/1122445.1122456}
\acmISBN{978-1-4503-9999-9/18/06}

\usepackage{setspace}

\usepackage{mathtools}

\usepackage{amsmath}
\usepackage{color}
\usepackage{graphicx}
\usepackage{threeparttable}
\usepackage[linesnumbered, ruled, vlined]{algorithm2e}
\usepackage{amsmath}
\usepackage{amssymb}
\usepackage{amsthm}
\usepackage{mathtools}

\usepackage[normalem]{ulem}

\newtheorem{feature}{Feature}[section]

\usepackage[noend]{algpseudocode}
\usepackage{indentfirst}

\makeatletter
\newcommand{\rmnum}[1]{\romannumeral #1}
\newcommand{\Rmnum}[1]{\expandafter\@slowromancap\romannumeral #1@}
\makeatother

\usepackage{array}
\newcolumntype{C}[1]{>{\centering\arraybackslash}m{#1}}

\begin{document}

%
\title[A Framework for Allocating Server Time to Spot and On-demand Services]{A Framework for Allocating Server Time to Spot and On-demand Services in Cloud Computing}

%
\author{Xiaohu Wu}
\affiliation{%
  \institution{Nanyang Technological University}
  \country{Singapore}
}
\email{xiaohu.wu@ntu.edu.sg}

\author{Francesco De Pellegrini}
\affiliation{%
  \institution{University of Avignon}
  \city{Avignon}
  \country{France}
  }
\email{francesco.de-pellegrini@univ-avignon.fr}

\author{Guanyu Gao}
\affiliation{%
  \institution{Nanyang Technological University}
  \country{Singapore}
}
\email{ggao001@ntu.edu.sg}

\author{Giuliano Casale}
\affiliation{%
 \institution{Imperial College London}
 \city{London}
 \country{United Kingdom}}
\email{g.casale@imperial.ac.uk}

%
\renewcommand{\shortauthors}{Xiaohu Wu, Francesco De Pellegrini, Guanyu Gao, and Giuliano Casale}

%
\begin{abstract}
Cloud computing delivers value to users by facilitating their access to servers in periods when their need arises. An approach is to provide both on-demand and spot services on shared servers. The former allows users to access servers on demand at a fixed price and users occupy different periods of servers. The latter allows users to bid for the remaining unoccupied periods via dynamic pricing; however, without appropriate design, such periods may be arbitrarily small since on-demand users arrive randomly. This is also the current service model adopted by Amazon Elastic Cloud Compute. In this paper, we provide the first integral framework for sharing the time of servers between on-demand and spot services while optimally pricing spot service. It guarantees that on-demand users can get served quickly while spot users can stably utilize servers for a properly long period once accepted, which is a key feature to make both on-demand and spot services accessible. Simulation results show that, by complementing the on-demand market with a spot market, a cloud provider can improve revenue by up to 461.5\%. The framework is designed under assumptions which are met in real environments. It is a new tool that other cloud operators can use to quantify the advantage of a hybrid spot and on-demand service, eventually making the case for operating such service model in their own infrastructures.
\end{abstract}

%
%

\ccsdesc[500]{Networks~Cloud computing}
\ccsdesc[500]{Networks~network economics}
\ccsdesc[500]{Computing methodologies~Model development and analysis}
\ccsdesc[100]{Mathematics of computing~Queueing theory}

%
\keywords{cloud computing, spot and on-demand services, time allocation, pricing}

%

%
\maketitle

\section{Introduction}
\label{sec.pricing-goals}

The global cloud Infrastructure-as-a-Service (IaaS) market grew to \$34.6 billion in 2017, and is projected to increase to \$71.6 billion in 2020 \cite{gartner-rp-1}.
IaaS delivers value to users by facilitating their access to servers or virtual machines: users can rent servers from cloud service providers (CSPs) whenever their need arises.
From a CSP's perspective, an important question is what forms of services should be offered to attract more users and achieve high resource efficiency.
One example is Amazon Elastic Cloud Compute (EC2), the leading CSP, accounting for 51.8\% of the global market share in 2017 \cite{gartner-rp-1}. There are two types of service: on-demand and spot instances (i.e., virtual machines) \cite{amazon-pricing,Ben-Yehuda13a}.
The former are offered at a fixed price and users pay only for the period in which instances are consumed. They may be idle at times, and such states are sold in the form of spot instances to improve resource efficiency. Users can bid prices for spot instances that will be offered and run as long as their bid is above the spot price; spot users are billed by the spot price that is usually lower than the price of on-demand instances (i.e., on-demand price). So far, numerous works have been done for users to utilize these instances cost-efficiently \cite{Kumar,Li16a}.


Offering users on-demand and spot services is an interesting option for CSPs. To understand the internal process, two intertwined aspects need to be controlled at once: (\rmnum{1}) the sharing of server time among on-demand and spot users and (\rmnum{2}) the pricing of spot instances.
A framework covering these aspects is needed in order to offer such services. Existing works \cite{Abhishek12,Dierks16a} focus on the case that the spot and on-demand markets are isolated from each other \cite{Devanur17}.
The pricing of spot instances is studied in \cite{Wang13a} where a demand curve is used to describe the relation between the spot price and the number of accepted bids \cite{Wang13a}. However, the curve is unknown in reality \cite{Besbes09a}. In \cite{Abhishek12,Dierks16a,Wang13a}, the connection of the two markets is ignored; however, the spot market is supposed to utilize the idle servers of on-demand market. Such idleness affects the capacity of spot market to accept bids and the specific schemes for assigning bids to servers.


While sharing servers among users, the model needs a feature: on-demand users can get served shortly upon arrivals, while a spot user can stably utilize a server for a properly long time once its bid is accepted. It is challenging to obtain this feature. In particular, each on-demand user requests to occupy servers in a specific period; spot users have lower priority to access servers and bid prices to utilize the unoccupied periods. After a spot user's bid is accepted and assigned to a server, the server may be preempted at any time after the assignment, since on-demand users arrive randomly. What's worse, for each accepted bid, there may be a process of loading virtual machine image to make the server ready for use; assume it takes $k^{\prime}$ minutes, e.g., $k^{\prime}$ approximates 3 \cite{Mao12a,Razavi13a}. So, the time that the server dedicates to an accepted bid needs to be larger than $k^{\prime}$ minutes so that spot users can get some effective utilization time of servers.
We will use a discrete-time model to coordinate the assignment of all users' requests to servers; thus, time is divided into slots and each contains $k$ minutes. These requests may arrive at any time in a slot; the assignment is performed at the beginning of the next one and it is the only action that changes the server state (i.e., occupied or idle). So, once a bid is assigned at slot $t$, it can stably utilize a server along the slot, without being interfered by the high priority of on-demand requests; here, the effective utilization time is $k-k^{\prime}$ minutes and spot users may hope that $k-k^{\prime}$ is properly large, e.g., more than ten minutes. Contradictorily, the slot duration bounds the waiting time of an on-demand request, since it needs to wait for up to $k$ minutes to get assigned and served; so, the length $k$ of a slot cannot be large.

Based on the observation above, we further propose a parallelized service model able to exploit the trade-off between preemption of spot instances for on-demand requests and service persistence for spot requests.
In particular, all servers are divided into $b$ groups; for the $i$-th group where $i\in [1, b]$, the assignment of requests occurs at the beginning of slot $t=h\cdot b+i$ where $h=0, 1, 2, \cdots$; these requests arrive in the last slot $t-1$. Consequently, the server state of each group can keep constant for $b$ slots; at any slot, the arriving jobs will be connected and assigned to a specific group. So, in order to guarantee that the effective server utilization time of an accepted bid is positive, the requirement is $b\cdot k>k^{\prime}$. A parallelized model allows us to set $k$ to a value small enough and set $b$ to a value properly large; here $k$ can be smaller than $k^{\prime}$. As a result, an on-demand request will get served shortly (within $k$ minutes). After a bid is accepted, the effective server utilization time could also be large, i.e., $b\cdot k-k^{\prime}$ minutes.


\vspace{0.15em}\noindent\textbf{Main Results.} Servers are shared among on-demand and spot users that arrive randomly; the former have higher priority to access servers at a fixed price. Spot users bid prices for the periods unoccupied by on-demand users.
A key feature that makes such services accessible is that, on-demand users can get served within a short time upon arrivals, while a spot user can stably utilize a server for a properly long time once its bid is accepted. In this paper, we propose a discrete-time framework that has such a feature for allocating the time of servers between on-demand and spot services. The framework presents an integral process inside the system to serve the arriving on-demand and spot users: (\rmnum{1}) assign the requests of on-demand users to servers, (\rmnum{2}) determine the optimal spot price, (\rmnum{3}) decide which spot users' bids are accepted, with the idle states of servers after the assignment of on-demand requests, and (\rmnum{4}) assign the requests of spot users to servers; these actions occur at the beginning of every time slot.
The framework itself does not rely on any impractical assumption and can guide other CSPs to operate on-demand and spot services in their own infrastructures. Simulations show that the CSP's revenue could be improved by up to 461.5\% while the resource utilization could be improved by up to 725.0\%, compared with the case of only providing on-demand service.



The rest of this paper is organized as follows. In Section~\ref{sec.preliminary-related-works}, we further explain the on-demand and spot services and introduce the related work. We propose the model and schemes for managing the requests of spot and on-demand users in Section~\ref{sec.basic-resource-management}; the optimal pricing of spot instances is given in Section~\ref{sec.optimal-spot-pricing}. To improve quality of services, we further propose an extended framework in Section~\ref{sec.extended-model}. Simulations are done to show the efficiency of spot pricing in Section~\ref{sec.performance-evaluation}. Finally, we conclude the paper in Section~\ref{sec.conclusion}.

\section{Preliminaries, and Related Work}
\label{sec.preliminary-related-works}

Before we formally introduce the proposed framework, we provide an overview of the state of the art on spot and on-demand services
and introduce the main related works.

\subsection{On-demand and Spot Services}
\label{sec.preliminary}

Amazon EC2 is the current reference cloud service provider based on spot and on-demand instances. Instances are virtual machines\footnote{In this paper, we use the terms "instances", "servers", and "virtual machines" interchangeably.} and can have different configurations of CPU, cache and disks. Instances of the same configuration have the same price and form a single market in Amazon EC2 \cite{amazon-pricing}. In this paper we also consider such a homogeneous case. On-demand instances are always available at a fixed unit price and each instance's price is charged on an hourly basis. Even if partial hour of on-demand instances is consumed, the tenant will be charged the fee of the entire hour. Spot instances are with uncertain availability and their price (termed as "spot price") fluctuates over time \cite{Ben-Yehuda13a}. Every user can bid a price for spot instances and they will be granted to users only if the bid price is not below the spot price. The bid price is the maximum price that the spot user can accept to pay for the spot instances. Once the spot price exceeds the user's bid price, its spot instances will get lost and terminated by Amazon EC2. Users will be charged according to the spot prices.

From a CSP's perspective, spot instances render available the computing capacity unused in the on-demand market and allow for some discount compared with the on-demand price. This permits increasing the CSP's gain in terms of revenue and user satisfiability. First, users with different delay requirements can be satisfied economically by the two types of instances  \cite{Jain14a,Wu17a}. Latency-critical users can get service quickly by specifying a period in which to utilize instances. Delay-tolerant users can first utilize spot instances at lower prices but with uncertain availability; in case that users do not get enough instances in a long period, they can turn to on-demand instances to accelerate processing their jobs. Next, the periods of servers unoccupied by on-demand users correspond to the idle states of on-demand instances. In the Amazon EC2 service, spot users can bid prices to utilize these states, i.e., spot instances. The pricing mechanism is not fully disclosed; however, the spot price is claimed to be set through a uniform price, sealed-bid, market-driven auction \cite{Ben-Yehuda13a}: "uniform price" means all bidders pay the same price (i.e., spot price); "sealed-bid" means a user does not know the bids of the other users; "market-driven" means the spot price fluctuates based on the supply and demand of available unused EC2 capacity but is updated regularly.

\subsection{Related Work}
\label{sec.related-works}

Many works have focused on characterizing spot prices over time and understanding the spot pricing scheme \cite{Kumar,Li16a};
the most relevant are \cite{Ben-Yehuda13a,Wang13a,Zheng15}. 
Agmon Ben-Yehuda {\em et al.} analyze the time series of spot prices in different regions and define the availability of spot instances as a function of the bid price, i.e., the probability that a user successfully gets spot instances under an arbitrary bid price \cite{Ben-Yehuda13a}. The authors showed that the functional curves of different types of instances in 4 regions share the same shape. Further, they conclude that spot prices in Amazon EC2 are usually drawn from a tight, fixed range of prices and are not driven by the relation of supply and demand as claimed by Amazon EC2. The authors claim that the use of such a pricing scheme can create an impression of false activity (demand and supply changes) and mask times of low demand and price inactivity, thus possibly driving up the CSP's stock.

Wang {\em et al.} study the optimal pricing of spot instances \cite{Wang13a}. They assume perfect knowledge of the demand curve describing the relation of demand and price at every slot $t$, i.e., the number $N_{t}$ of bids accepted and served under every possible spot price. Further, Lyapunov optimization is applied to derive the optimal price of spot instances by assuming that the total number $L_{t}$ of bids is kept finite at every slot $t$; here the $L_{t}$ bids contain the bids that both newly arrive at $t$ and all bids that arrived at the previous slots but have not been served so far. However, the demand curve is unknown in reality \cite{Besbes09a}; an additional complication is that the current cloud market is still rapidly growing and unstable \cite{gartner-rp-1}.

Zheng {\em et al.} derive the cost-optimal bid price for users to utilize spot instances, based on an estimated distribution of the past spot prices in Amazon EC2 \cite{Zheng15}. In particular, similar to \cite{Wang13a}, Lyapunov optimization is used to derive the relation between the number of bid arrivals $\Lambda_{t}$ and the spot price $\pi_{t}$ at every slot $t$; as a result, by assuming the bid arrivals follow a simple distribution (e.g., exponential), the more complex distribution of spot prices could be approximated analytically. Other assumptions adopted include (\rmnum{1}) the users' bids at every slot $t$ follow a uniform distribution over $[\underline{\pi},\, \overline{\pi}]$, and (\rmnum{2}) at the end of each slot, the proportion of the accepted bids that are finished is a constant. The first assumption enables simply deriving the expected number of the bids accepted at $t$, i.e., the fraction of bids whose prices are not below the spot price.

In \cite{Wang13a,Zheng15}, due to their assumptions, the number of bids accepted at every slot does not need to rely on the idle state of the on-demand market. However, the spot market's capacity to accept bids is supposed to be determined by such idleness. Their models cannot account for the actual preemption scheme used for sharing server time among on-demand and spot users. 
Furthermore, in order to apply the Lyapunov optimization technique, an underlying assumption in \cite{Zheng15,Wang13a} is that all bids submitted by users will be finally accepted and served. However, this is not what Amazon EC2 promises to its spot users, and it only provides best-effort services.

Abhishek {\em et al.} and Dierks {\em et al.} analyze the performance of a hybrid spot and on-demand market using queuing theory and game theory in continuous time \cite{Abhishek12,Dierks16a}. More recently, in \cite{Dierks16a}, the servers are separated into two parts, respectively serving on-demand and spot jobs. Jobs of users have diverse values and sensitivities to delay. On-demand market has a higher price but guarantees a negligible delay; in spot market, the lower a user's bid, the larger the delay of completing its jobs. Users aim to maximize their surplus and make a choice on which type of instances to use. The authors show that offering a spot market can increase the profit of a CSP. In spite of the technical merits of \cite{Abhishek12,Dierks16a}, an issue in these works is that in the proposed schemes the idle servers in the on-demand market cannot be sold as spot instances, whereas making use of such idle instances is one of the main attractions of the spot market \cite{Devanur17}.

Furthermore, inspired by the dynamic pricing of Amazon EC2, there are also many works that apply the auction and mechanism design theory to cloud pricing \cite{Zhang17a,Shi17a,Shi16a,Fu14a,Shi14a,Zhang14a,Shi14b,Wang13b,Zhang18b,Zhou16a,Zhou18a,Zhang16a}; the most relevant ones are \cite{Jain15a,Wu15a,Azar15a}. In those frameworks, jobs truthfully report their values and latency requirements to the CSP; the CSP chooses a subset of jobs to maximize the social welfare, and process the chosen jobs on multiple machines with its capacity constraint. Unlike the model of this paper, Wu and De Pellegrini \cite{Wu17b} analyze the performance of a type of QoS-differentiated pricing in cloud computing via an analytical approach. The CSP offers multiple QoS classes: the jobs of each class will be completed with a finite waiting time. Also, the smaller the waiting time, the higher the unit price.
The CSP's servers are divided into several groups, each processing the jobs of the same class. The authors give the optimal price, the optimal rate of accepting jobs, and the minimal number of servers needed for each QoS class, thus deriving the performance of the whole system.


\section{Managing On-demand and Spot Services}
\label{sec.basic-resource-management}

In this section, we propose a baseline model for colocating on-demand and spot jobs on servers. The key notation is summarized in Table~\ref{table}. The model refers to several schemes to (\rmnum{1}) assign on-demand jobs to servers where every server will dedicate a specified period to the assigned job, (\rmnum{2}) decide which bids to accept under an arbitrary spot price, and (\rmnum{3}) assign spot jobs to servers. As discussed in Section~\ref{sec.preliminary}, we assume that there are $m$ homogeneous servers.


\begin{table*}
	\centering
		\caption{Key Notation}
		\begin{tabular}{| C{1.3cm} | C{11.5cm} |}   
		
        	\hline

			Symbols & Meaning\\ \hline

             $m$    &  the total number of servers/instances \\ \hline

			$M_{t}$   &   the number of servers idle at slot $t$ for spot market \\ \hline


			$\mathcal{A}_{t}$   &   the set of all bids (or spot jobs) at $t$ where $\mathcal{A}_{t}=\mathcal{J}_{t}^{\prime}\cup\mathcal{J}_{t}^{\prime\prime}$ \\ \hline

			$\pi_{t}$   &   the spot price at $t$ \\ \hline


			$N_{t}$   &   the number of bids accepted at $t$ \\ \hline


            $\mathcal{J}_{t}^{\prime}$   &   the bids of such users who bid successfully at $t-1$ and continues bidding at $t$\\ \hline

            $\mathcal{J}_{t}^{\prime\prime}$  &  the bids that are newly submitted in the period of $[(t-1)\cdot k, t\cdot k)$ \\ \hline

            $\hat{\mathcal{J}}_{t}^{\prime}$  &  the bids in $\mathcal{J}_{t}^{\prime}$ that are also accepted at $t$\\ \hline

            $\hat{\mathcal{J}}_{t}^{\prime\prime}$  &  the bids in $\mathcal{J}_{t}^{\prime\prime}$ accepted at $t$\\ \hline

            $k$  &  a slot contains $k$ minutes \\ \hline

            $k^{\prime}$  &  the time spent on loading/migrating virtual machine images (VMIs) is $k^{\prime}$ minutes \\ \hline      	

            $\beta$  &  $k^{\prime}/k$ \\ \hline

            $\hat{\mathcal{J}}_{t,1}^{\prime}$    &  all bids of $\hat{\mathcal{J}}_{t}^{\prime}$ that cause the operation of migrating VMI at the beginning of $t$\\ \hline

            $\hat{\mathcal{J}}_{t,2}^{\prime}$    &  $\hat{\mathcal{J}}_{t}^{\prime}-\hat{\mathcal{J}}_{t,1}^{\prime}$, in which each bid is assigned to the same server at $t-1$ and $t$\\ \hline

             $f_{t}$    &  $|\hat{\mathcal{J}}_{t}^{\prime\prime}\cup \hat{\mathcal{J}}_{t,1}^{\prime}|$, i.e., the total number of the operations of loading/migrating VMI\\ \hline
            		
		\end{tabular}
		\label{table}
\end{table*}

\subsection{Model for Colocating On-demand and Spot Jobs}
\label{sec.model_colocating}

On-demand jobs have higher priority than spot jobs to access servers. On demand jobs arrive randomly and spot jobs may be preempted arbitrarily. However, the resource sharing model needs to accommodate two seemingly conflicting requirements:
\begin{description}
\item [(\rmnum{1}) Immediacy of on-demand service.] Upon arrival of an on-demand job, it can get served in a short time;
\item [(\rmnum{2}) Persistence of spot service.] After a spot job enters service, it will not be preempted shortly by on-demand jobs; ideally, the colocation scheme needs to guarantee a minimum execution time for spot jobs.
\end{description}
We first need to provide a discrete-time model able to describe the assignment of jobs to servers and express the above tradeoff. To this aim, time is divided into slots and the assignment occurs at the beginning of a slot. The slot duration is $k$ minutes, and the $t$-th slot corresponds to the period of $[t\cdot k, (t+1)\cdot k)$ where $t=0, 1, 2, \cdots$.

\begin{figure}[t]
\begin{center}
  \includegraphics[width=2.65in]{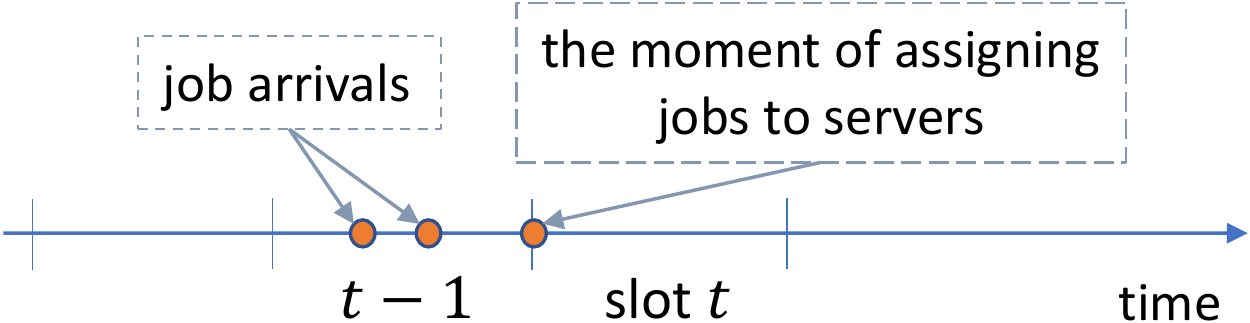}
  \caption{Discrete-time service model where $t=1, 2, \cdots$.}\label{Fig.discrete-time-model}
\end{center}
\end{figure}

From the time point 0, jobs begin to arrive; the initial slot is slot 0 during which all servers are idle. All jobs that arrive in the period of slot $t-1$ will be available at the beginning of slot $t$ where $t=1, 2, \cdots$. Among these jobs, the accepted jobs will get served and dispatched to servers at the beginning of slot $t$. For convenience, we will simply say from a system administrator's perspective that the arrival time of these jobs is $t$. The discrete-time model is also illustrated in Fig.~\ref{Fig.discrete-time-model}. In such a model, the action that changes the states of servers is the assignment of jobs to servers and it only occurs at the beginning of every slot, and their states remain constant in the period of each slot. As a result, if a spot job is executed on some server, it will not be preempted by the high priority of on-demand jobs during a slot and the minimum time that the server dedicates to it is $k$ minutes. Formally, the discrete-time model provides the on-demand and spot services with the following properties that respectively quantify the immediacy of on-demand service and the persistence of spot service.

\begin{feature}\label{feature-11}
On-demand users can be latency-critical. Upon arrival, they need to wait for up to $k$ minutes to get served.
\end{feature}

\begin{feature}\label{feature-12}
Spot users are usually delay-tolerant. Once bid successfully at a slot and assigned to some servers, it is guaranteed that they can persistently get served for $k$ minutes.
\end{feature}

Only after providing the model for sharing servers, we shall be able to figure out the framework for assigning jobs and accepting bids. The conceptual framework is illustrated in Fig.~\ref{Fig.framework-1} and will be elaborated in the rest of this section. In this paper, although all actions of assigning/dispatching, accepting, and pricing jobs occur at the beginning of every slot, we will simply say that they occur at slot $t$ for convenience of exposition.

\begin{figure}[t]
\begin{center}
  \includegraphics[width=3.65in]{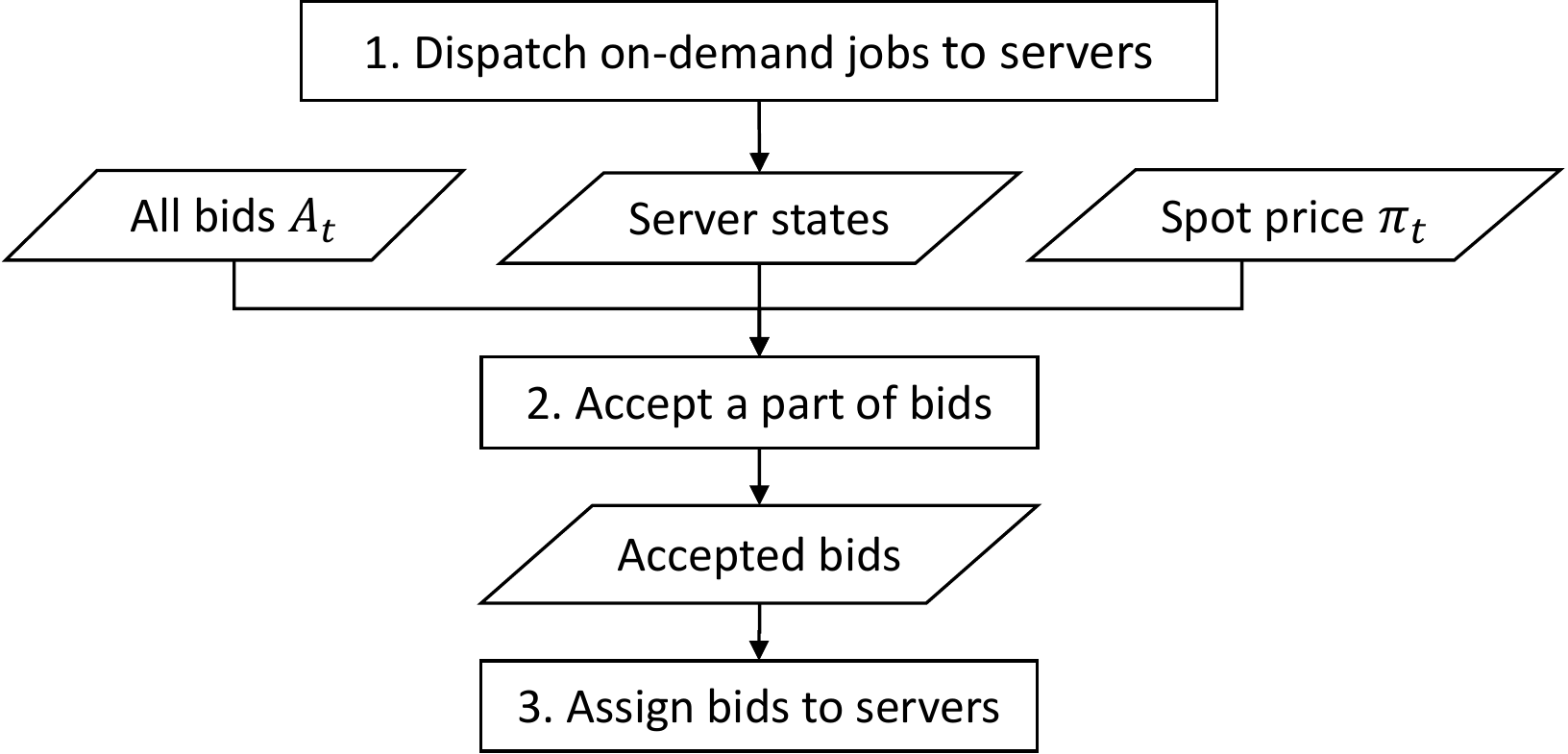}
  \caption{A flowchart for assigning jobs and accepting bids at slot $t$ where $t=1, 2, \cdots$: a rectangle represents a process that executes some operations while a parallelogram represents the input or output of a process.}\label{Fig.framework-1}
\end{center}
\end{figure}



\subsection{Dispatching High Priority of On-demand Jobs}
\label{sec.on-demand-jobs}

Each on-demand user requests a time slot interval $[a_{j}, d_{j}]$ in which it can occupy a server to execute its workload where $a_{j}$ and $d_{j}$ are positive integers. We refer to such a request as an on-demand job $j$, and its arrival time, deadline and size are $a_{j}$, $d_{j}$, and $s_{j}=d_{j}-a_{j}+1$ respectively. At every slot $t$, on-demand jobs with $a_{j}=t$ are dispatched to one of the $m$ servers under some policy. In IaaS services, examples of the commonly used policies include (\rmnum{1}) {\em Random}: for every job $j$, choose one of the $m$ servers with the probability $\frac{1}{m}$ and assign it to this server \cite{Zheng16a,Rasley16a}, (\rmnum{2}) {\em Round-Robin} (RR): jobs are assigned to servers in a cyclical fashion with the $j$-th job being assigned to the $l$-th server where $l=j\, mod\, m$ \cite{Wang14a}, (\rmnum{3}) {\em Power of Two Choices} (PTC): for every job $j$, randomly choose two servers, probe them, and, assign it to the server with less queued jobs \cite{Ousterhout13a,Mitzenmacher01a}.
The RR and Random policies are simple to implement and have similar performance. The PTC policy is more advanced but recently has been applied to the cluster management practice \cite{Ousterhout13a}; it can achieve a higher utilization of servers \cite{Mitzenmacher01a}. Once the job $j$ is dispatched to a server, the server will be occupied by the on-demand job owner during the period $[a_{j}, d_{j}]$, i.e., from the beginning of slot $a_{j}$ until the end of slot $d_{j}$.



On-demand instances are charged a fixed price; their users can be delay-sensitive and have no willingness to tolerate queuing delay: the current practice to guarantee quick delivery of on-demand instances to users is overprovisioning servers for on-demand jobs. As a result, while processing on-demand jobs, many servers actually remain unoccupied in the long run. The idleness in the on-demand market will be shown in Section~\ref{sec.on-demand-idleness} through experiments and available theoretical results, e.g., the load of the on-demand market corresponds to more than 85\% servers in idle mode.
After the job assignment at $t$, a server is either occupied by an on-demand job or idle in the entire period of slot $t$. We use $M_{t}$ ({\em resp.} $\overline{M}_{t}$) to denote the number of idle ({\em resp.} occupied) servers in the period of $t$, where $\overline{M}_{t}+M_{t}=m$. The idleness of on-demand market brings the necessity of introducing spot instances into cloud market \cite{Devanur17}; in particular, to be economically efficient, what we can do at every slot $t$ is as follows:
\begin{itemize}
\item sell the idle states of on-demand market in the period of slot $t$ in the form of spot instances; they will be accessed by spot users via bidding, and their amount is $M_{t}$.
\item however, on-demand jobs have higher priority to access servers, i.e., once new on-demand jobs arrive at slot $t+1$, the instances assigned to spot users at $t$ may be preempted arbitrarily at the beginning of $t+1$ to serve on-demand jobs if necessary.
\end{itemize}
The additional sales of spot instances are supposed to improve the overall resource efficiency of cloud market, in contrast to a pure on-demand market.

\subsection{Admission Control of Spot Jobs via Pricing}
\label{sec.pricing}

At every slot $t$ where $t=1, 2, \cdots$, there are $A_{t}$ users who bid prices to utilize spot instances in the period of slot $t$; they are usually delay-tolerant. The set of these users' bids is denoted by $\mathcal{A}_{t}$. We assume without loss of generality that each user bids for one spot instance. Once the bid\footnote{In this paper, each bid at a slot corresponds to a spot job who aims to utilize an instance for one slot once accepted; we shall use the terms "bids" and "spot jobs" interchangeably.} of a user is accepted, it gets allocated one spot instance at $t$. Let $\pi_{t}$ denote the spot price at $t$. Following the spirit of spot pricing described in Section~\ref{sec.preliminary}, $\pi_{t}$ is a control parameter: only the bids whose prices are not below $\pi_{t}$ are accepted. The number of idle on-demand instances $M_{t}$ defines the capacity of spot market, i.e., the number of bids accepted at $t$ cannot exceed $M_{t}$.

\begin{figure}[t]
\begin{center}
  \includegraphics[width=3.25in]{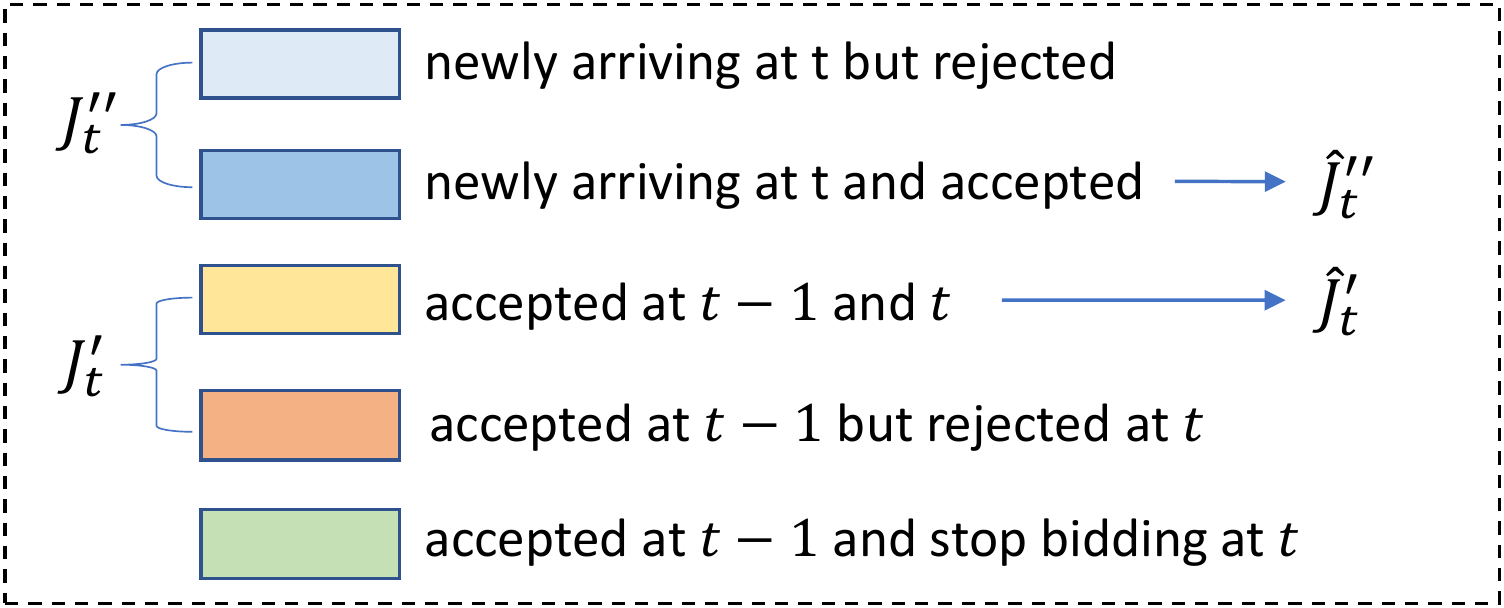}
  \caption{Color chart for the classification of spot jobs/bids available at $t$ when $t\geq 2$: we note that, when $t=1$, there are only bids of $\mathcal{J}_{t}^{\prime\prime}$ and no other bids, e.g., $\mathcal{J}_{t}^{\prime}=\emptyset$.}\label{Fig.classification-spot-jobs}
\end{center}
\end{figure}

Before defining a scheme for accepting bids, we first classify the bids available at $t$ from a system administrator's perspective. This will help us describe which bids are accepted and understand how to assign these bids to servers in the next subsection. A physical spot user may bid once or several times in order to utilize spot instances at one or several slots; every time, its bid may be accepted or rejected. It stops bidding when either its entire job is completed or it wants to use other resources (e.g., on-demand instances) to complete its job. If a spot user's bid is accepted at several consecutive slots, its spot job may be assigned to and executed on the same server in this period; however, once its bid is rejected at a slot and accepted at a later slot $t$, it can be viewed as a user that newly arrive at $t$ since it is not associated with and can be assigned to any server at $t$.

Thus, we will classify the bids of spot users at $t$ according to their bidding behaviour at the adjacent slot $t-1$, and there are two types of bids where $\mathcal{A}_{t}=\mathcal{J}_{t}^{\prime}\cup \mathcal{J}_{t}^{\prime\prime}$:
\begin{itemize}
\item [(\rmnum{1})] $\mathcal{J}_{t}^{\prime}$: the bids of the spot users whose bids were also accepted at slot $t-1$;
\item [(\rmnum{2})] $\mathcal{J}_{t}^{\prime\prime}$: the bids of the spot users who newly arrive in the period of $[(t-1)\cdot k, t\cdot k)$.
\end{itemize}
Specially, if $t=1$, $\mathcal{J}_{t}^{\prime}=\emptyset$ since there are no bids accepted at the initial slot $t-1=0$ during which there are newly arriving bids alone. For each bid of $\mathcal{J}_{t}^{\prime}$, its owner also bids at $t-1$; as seen later, at $t-1$ and $t$, the two bids will be assigned to the same server when there is no interference from the high priority of on-demand jobs. Finally, $\mathcal{A}_{t}$ denotes all bids available at $t$: $\mathcal{J}_{t}^{\prime}$ is illustrated by the orange and gold rectangles in Fig.~\ref{Fig.classification-spot-jobs}; $\mathcal{J}_{t}^{\prime\prime}$ is illustrated by the light and heavy blue rectangles.


\vspace{0.25em}\noindent\textbf{Acceptance of spot jobs.} The decision of accepting bids is made at the beginning of slot $t$. To describe the bids accepted at a slot, we define a function as follows. Generally, let $\mathcal{J}=\{1, 2, \cdots, J\}$ denote a set of users who submit bids where $J=|\mathcal{J}|$, and $\mathcal{V}=\{v_{1}, v_{2}, \cdots, v_{J}\}$ denote the set of their bid prices where $v_{1}\geq v_{2}\geq \cdots\geq v_{J}$. Let $\alpha$ denote a non-negative real number; we define a function as follows:
\begin{equation}
F(\alpha, \mathcal{J}) = \{i \,|\, v_{i}\geq \alpha, i\in\mathcal{J} \}.
\end{equation}
$F(\alpha,\, \mathcal{J})$ denotes all users whose bid prices are no smaller than $\alpha$. The number of the bids whose prices are not below $\pi_{t}$ is $|F(\pi_{t},\, \mathcal{A}_{t})|$. The procedure for determining which bids are accepted is presented in Algorithm~\ref{procedure-accepting-bids}; at slot $t$, the number of accepted bids is the minimum of $M_{t}$ and $|F(\pi_{t},\, \mathcal{A}_{t})|$, i.e.,
\begin{equation}\label{equa-accepted-bids}
N_{t} = \min\{M_{t},\, |F(\pi_{t},\, \mathcal{A}_{t})|\}.
\end{equation}

\begin{algorithm}[t]
	\SetKwInOut{Input}{Input}
	\SetKwInOut{Output}{Output}	
	
	\BlankLine	


\If{$|F(\pi_{t},\, \mathcal{A}_{t})|\leq M_{t}$}{

      accept $|F(\pi_{t},\, \mathcal{A}_{t})|$ bids of $\mathcal{A}_{t}$ with the highest bid prices\;
}
\Else{
      accept $M_{t}$ bids of $\mathcal{A}_{t}$ with the highest bid prices\;
}

	\caption{Determination of bids accepted at $t$\label{procedure-accepting-bids}}
\end{algorithm}

Also, we note that $M_{t}$ and $\mathcal{A}_{t}$ are observable at the beginning of slot $t$ and thus known by CSP. We denote by $\hat{\mathcal{J}}_{t}^{\prime}$ (resp. $\hat{\mathcal{J}}_{t}^{\prime\prime}$) the accepted bids of $\mathcal{J}_{t}^{\prime}$ (resp. $\mathcal{J}_{t}^{\prime\prime}$).
Here, $\hat{\mathcal{J}}_{t}^{\prime}$ and $\hat{\mathcal{J}}_{t}^{\prime\prime}$ are illustrated by gold and heavy blue rectangles in Fig.~\ref{Fig.classification-spot-jobs}. In the rest of this section we shall refer to the color chart of  Fig.~\ref{Fig.classification-spot-jobs} to support the description of job assignment with a graphical representation.



\subsection{Assignment of Spot Jobs to Servers}
\label{sec.dispatch}

\begin{figure}[t]
\begin{center}
  \includegraphics[width=4.45in]{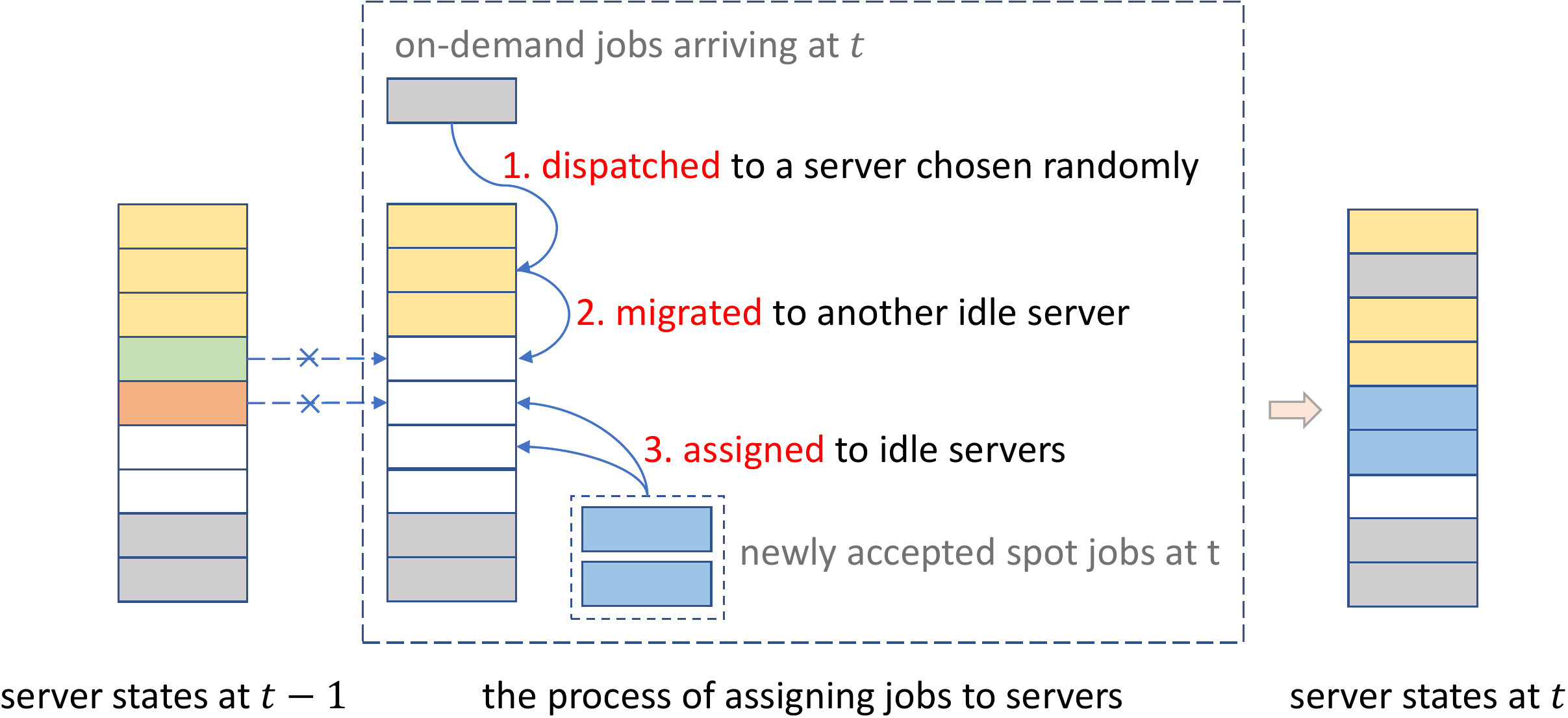}
  \caption{Assignment of on-demand and spot jobs to servers at the beginning of slot $t$ when $t\geq 2$: a colored rectangle denotes either a job of some type or the state of a server (i.e., a server occupied by this type of jobs); the meaning of different colors is partially summarized in Fig.~\ref{Fig.classification-spot-jobs}; a blank rectangle denotes an idle server.}\label{Fig.managing-spot-on-demand-instances}
\end{center}
\end{figure}

At every slot $t=1, 2, \cdots$, all the arriving on-demand jobs will be accepted but only a part of bids may be accepted. On-demand jobs have higher priority; when dispatching them to servers, the existence of spot jobs is ignored as if there is a pure on-demand market. The scheme for assigning jobs to servers is as follows:
\begin{description}
\item [Step 1.] On-demand jobs are dispatched to servers using some policy such as "Random", "RR" or "PTC".
\item [Step 2.] In the case that $t\geq 2$, if an on-demand job is dispatched to a server that was executing a spot job in $\hat{\mathcal{J}}_{t}^{\prime}$ in the period of slot $t-1$, this spot job is migrated to another idle server; the other jobs in $\hat{\mathcal{J}}_{t}^{\prime}$ are still executed on the same instances. In the case that $t=1$, go to the next step directly since $\hat{\mathcal{J}}_{t}^{\prime}=\emptyset$.
\item [Step 3.] the spot jobs of $\hat{\mathcal{J}}_{t}^{\prime\prime}$ are randomly dispatched to the remaining idle servers.
\end{description}
In the period of slot 0, all servers are in idle states. At the beginning of slot 1, there are only bids of $\hat{\mathcal{J}}_{t}^{\prime\prime}$ and no bids in $\hat{\mathcal{J}}_{t}^{\prime}$; these bids are assigned to servers after the assignment of on-demand jobs. When $t\geq 2$, a part of servers might be occupied by on-demand and spot jobs in the period of slot $t-1$; the process of assigning jobs to servers is also illustrated in Fig.~\ref{Fig.managing-spot-on-demand-instances}. We note that, when $t=1$,
the states of all servers at $t-1$ will be represented by blank rectangles and the second step in the process of assigning jobs to servers could be removed since nothing is executed.


\section{Optimal Pricing of Spot Instances}
\label{sec.optimal-spot-pricing}

In this section, we shall determine the optimal spot price to maximize the CSP's revenue at every slot $t$. We denote it by $\pi_{t}^{*}$: it is the spot price finally announced to users at $t$. It actually determines the subset of bids accepted by Algorithm~\ref{procedure-accepting-bids}. Afterwards, we will conclude Sec.~\ref{sec.basic-resource-management} and Sec.~\ref{sec.optimal-spot-pricing} by showing the whole framework for running spot and on-demand services.


As is formally shown latter, the optimal spot price may be the bid price of some user. This can be intuitively perceived by contradiction: if the value of $\pi_{t}^{*}$ is between two bid prices, the revenue of a CSP will be increased by resetting $\pi_{t}^{*}$ to the larger bid price, and doing so does not affect the acceptance of spot jobs. To determine which bid price can maximize the revenue, we need to characterize the CSP's revenue function under an arbitrary spot price $\pi_{t}$. Under $\pi_{t}$, the accepted bids $\hat{\mathcal{J}}_{t}^{\prime}\cup\hat{\mathcal{J}}_{t}^{\prime\prime}$ are determined by the framework in Section~\ref{sec.basic-resource-management} (illustrated in Fig.~\ref{Fig.framework-1}). The framework for pricing spot instances at $t$ is illustrated in Fig.~\ref{Fig.framework-2} and will be elaborated in the following.

\begin{figure}[t]
\begin{center}
  \includegraphics[width=4.95in]{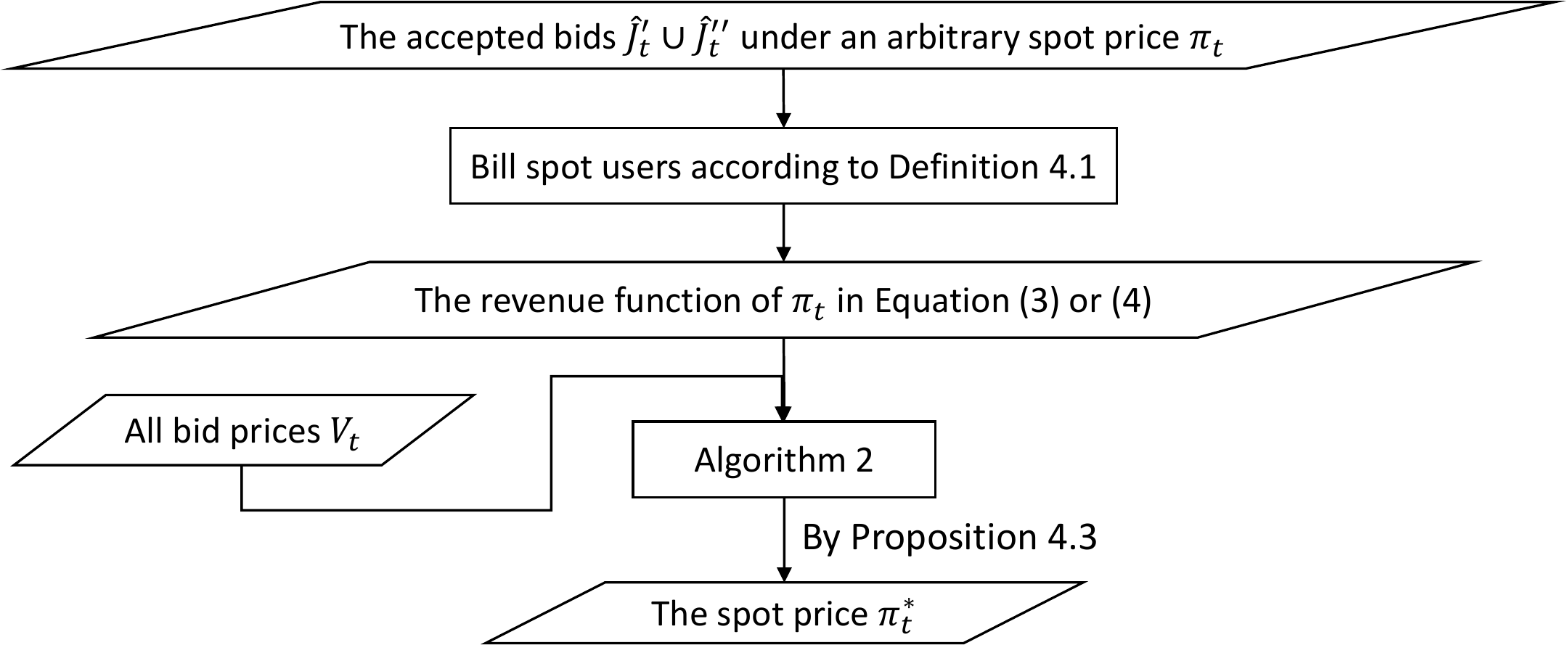}
  \caption{The optimal pricing of spot instances at the beginning of slot $t$.}\label{Fig.framework-2}
\end{center}
\end{figure}

\subsection{Billing and Revenue}

We first define the value that a CSP gets from each spot job accepted at $t$, i.e., the way of billing. The reference time interval for billing an on-demand user is made of $L$ consecutive slots. The price of utilizing an instance for one interval is $p$: we recall that if a user utilizes just a fraction of an interval, it is still charged the fee of the entire interval.

Conversely, for each accepted bid, we should define the billing way according to the effective utilization time of a spot instance. We observe that, although a spot instance is assigned by reserving a whole slot time, there may be yet some overhead. In fact, in some cases a process of startup/migration is needed: it takes some time during which the job is in service but inactive. The spot price at slot $t$ is $\pi_{t}$, and, to be comparable with the on-demand price, $\pi_{t}$ is the price of utilizing a spot instance for an interval of $L$ slots. However, we shall use a proper normalization, since $\pi_{t}$ is limited to represent the spot price of slot $t$. For example, if a spot user effectively utilizes an instance for one minute during the slot $t$, the actual charged price will be $\pi_{t}/(L\cdot k)$, where a slot contains $k$ minutes.


Now, we observe the period in which servers are effectively utilized by spot jobs after the job assignment at $t$. Recall the job assignment process in Section~\ref{sec.dispatch}. For spot jobs newly accepted $\hat{\mathcal{J}}_{t}^{\prime\prime}$, their virtual machine images (VMIs) need to be loaded to the assigned servers. As illustrated in Fig.~\ref{Fig.managing-spot-on-demand-instances}, some spot jobs of $\hat{\mathcal{J}}_{t}^{\prime}$ may need migration, denoted by $\hat{\mathcal{J}}_{t,1}^{\prime}$, and their VMIs also need be migrated to other servers. As shown in related studies \cite{Mao12a,Razavi13a}, for state of art technology, the process of loading or migrating VMIs takes about 3 minutes; generally, we denote the time consumed for this process by $k^{\prime}$ minutes. Thus, for the jobs of $\hat{\mathcal{J}}_{t}^{\prime\prime}\cup \hat{\mathcal{J}}_{t,1}^{\prime}$, although the whole period of a slot will be dedicated to them, only $k-k^{\prime}$ minutes are effectively utilized, with $k^{\prime}$ minutes not utilized for actual service. For the other spot jobs of $\hat{\mathcal{J}}_{t}^{\prime}$, they were ever executed on some servers at $t-1$ and will still be executed on the same servers at $t$; we denote these bids by $\hat{\mathcal{J}}_{t,2}^{\prime}$ where $\hat{\mathcal{J}}_{t,2}^{\prime}=\hat{\mathcal{J}}_{t}^{\prime}-\hat{\mathcal{J}}_{t,1}^{\prime}$. The bids of $\hat{\mathcal{J}}_{t,2}^{\prime}$ can effectively utilize the whole slot $t$. In order to guarantee that the effective server utilization time of every accepted bid is positive, the following relation needs to be satisfied:
\begin{center}
$k^{\prime}<k$.
\end{center}
Let $\beta=\frac{k^{\prime}}{k}$ where $\beta\in (0, 1)$, and every spot job will be charged for the period in which the servers are effectively utilized. Thus, we have the following definition.
\begin{definition}\label{def-billing-rule}
The way of billing a spot user at $t$ is as follows: (\rmnum{1}) every accepted bid in $\hat{\mathcal{J}}_{t}^{\prime\prime}\cup \hat{\mathcal{J}}_{t,1}^{\prime}$ is charged $(1-\beta)\cdot \frac{\pi_{t}}{L}$, and (\rmnum{2}) every accepted bid in $\hat{\mathcal{J}}_{t,2}^{\prime}$ is charged $\frac{\pi_{t}}{L}$.
\end{definition}

Now, we characterize the revenue function and show which system information is observable at the beginning of $t$. Given any spot price $\pi_{t}$, the accepted bids (i.e., $\hat{\mathcal{J}}_{t}^{\prime\prime}$, $\hat{\mathcal{J}}_{t,1}^{\prime}$, $\hat{\mathcal{J}}_{t,2}^{\prime}$) are determined by Algorithm~\ref{procedure-accepting-bids}, before which on-demand jobs have been assigned to servers. Also, we know the locations of the servers to which the bids accepted at $t-1$ are assigned, as illustrated in Fig.~\ref{Fig.managing-spot-on-demand-instances}. So, after the bids accepted at $t$ are determined, we could know the bids respectively in $\hat{\mathcal{J}}_{t,1}^{\prime}$, $\hat{\mathcal{J}}_{t,2}^{\prime}$ and $\hat{\mathcal{J}}_{t}^{\prime\prime}$. Let $f_{t}$ denote the total number of the accepted bids of $\hat{\mathcal{J}}_{t}^{\prime\prime}\cup \hat{\mathcal{J}}_{t,1}^{\prime}$; $f_{t}$ is also observable after the accepted bids are determined. The total number of the bids accepted at $t$ is $N_{t}$; so, $|\hat{\mathcal{J}}_{t,2}^{\prime}| = N_{t}-f_{t}$. The revenue of the spot market at $t$ is the sum of the charges of all accepted bids, denoted by $\mathcal{G}(t)$; with the billing policy in Definition~\ref{def-billing-rule}, we have
\begin{equation}\label{equa-gain}
\begin{split}
\mathcal{G}(t)  = (N_{t}-f_{t})\cdot \frac{\pi_{t}}{L} + f_{t}\cdot (1-\beta)\cdot \frac{\pi_{t}}{L}
                = \left(N_{t} - \beta\cdot f_{t}\right)\cdot \frac{\pi_{t}}{L}
\end{split}
\end{equation}
where $N_{t}$ is given in (\ref{equa-accepted-bids}), $L$ and $\beta$ are system parameters, and $f_{t}$ is observable.

\subsection{Pricing Decision}

Our decision-making problem is determining the optimal $\pi_{t}$ to maximize the spot market's revenue $\mathcal{G}(t)$ at slot $t$. By (\ref{equa-accepted-bids}), $N_{t}$ is a function of $\pi_{t}$, $\mathcal{A}_{t}$ and $M_{t}$; the revenue function $\mathcal{G}(t)$ in (\ref{equa-gain}) can be expressed as a function of the single variable $\pi_{t}$, i.e.,
\begin{align}\label{equa-gain-0}
\mathcal{G}(t) = \hat{\mathcal{G}}\left(\pi_{t},\, M_{t},\, \mathcal{A}_{t},\, f_{t}\right)
\end{align}
where parameters $M_{t}$, $\mathcal{A}_{t}$ and $f_{t}$ are observable at slot $t$.
Sort the spot jobs of $\mathcal{A}_{t}$ in the non-increasing order of their bid prices; let $v_{j}$ denote the bid price of the $j$-th spot job where
\begin{align}\label{bid-order}
v_{1}\geq v_{2}\geq \cdots\geq v_{A_{t}},
\end{align}
where $A_{t}=|\mathcal{A}_{t}|$. Let $v_{0}=v_{1}+1>v_{1}$, and $\mathcal{V}_{t}=\{v_{0}, v_{1}, \cdots, v_{A_{t}}\}$; then, we draw the following conclusion.

\begin{lemma}\label{lemma-optimal-price}
In order to maximize the revenue of spot market at slot $t$, the optimal spot price is such that $\pi_{t}^{*}\in \mathcal{V}_{t}$. Let $N_{t}^{*}$ denote the number of bids accepted at $t$ and we have $\pi_{t}^{*}=v_{N_{t}^{*}}$.
\end{lemma}
\begin{proof}
Suppose that in an optimal solution $N_{t}^{*}$ bids are accepted; as defined in Section~\ref{sec.pricing}, only the bids whose prices are no smaller than $\pi_{t}^{*}$ are possibly accepted, and these bids are the $N_{t}^{*}$ bids of $\mathcal{A}_{t}$ with the highest bid prices. If $N_{t}^{*}=0$, no bid is accepted at $t$ and $\pi_{t}$ can be an arbitrary value larger than $v_{1}$; so, $\pi_{t}^{*}$ can be set to $v_{0}$. If $N_{t}^{*}>0$, the optimal spot price $\pi_{t}^{*}\leq v_{N_{t}^{*}}$. Then, the CSP's revenue function $\mathcal{G}(t)$, given in (\ref{equa-gain}), is maximized when setting $\pi_{t}^{*}$ to the highest possible price, i.e. $\pi_{t}^{*}=v_{N_{t}^{*}}$. The reason for this is that, when $\pi_{t}\in [0, v_{N_{t}^{*}}]$ and the number of accepted bids $N_{t}$ is fixed and equals $N_{t}^{*}$, $\mathcal{G}(t)$ is an increasing function of $\pi_{t}$ since $N_{t} - \beta\cdot f_{t}>0$ where $N_{t}\geq f_{t}$ and $\beta\in (0, 1)$. Finally, the lemma holds.
\end{proof}


\begin{proposition}\label{proposi-optimal-price}
The optimal spot price $\pi_{t}^{*}$ at slot $t$ is as follows:
\begin{equation}\label{optimal-pricing-decision}
\pi_{t}^{*}\leftarrow \arg\max\limits_{\pi_{t}\in\mathcal{V}_{t}}{\hat{\mathcal{G}}(\pi_{t}, M_{t}, \mathcal{A}_{t}, f_{t})}.
\end{equation}
\end{proposition}
\begin{proof}
The optimal spot price at $t$ is some value in $\mathcal{V}_{t}$ under which $\mathcal{G}(t)$ achieves the maximal value, by Lemma~\ref{lemma-optimal-price}; hence, the proposition holds.
\end{proof}

\begin{algorithm}[t]
	\SetKwInOut{Input}{Input}
	\SetKwInOut{Output}{Output}	
	
	\BlankLine	

\setstretch{1.1}

$\pi^{\prime}\leftarrow 0$,\, $G^{\prime}\leftarrow 0$\;

\For{$i\leftarrow 0$ \KwTo $A_{t}$}{

    $N_{t} \leftarrow \min\{M_{t},\, |\mathcal{F}(v_{i}, \mathcal{A}_{t})|\}$\;

    $G \leftarrow N_{t}\cdot \frac{v_{i}}{L} - \beta\cdot f_{t}\cdot \frac{v_{i}}{L}$\tcp*{\footnotesize{the revenue of spot market if the spot price at $t$ is $v_{i}$}}

    \If{$G^{\prime} < G$}{
        $G^{\prime}\leftarrow G$,\, $\pi^{\prime}\leftarrow v_{i}$\;
    }
}
$\pi_{t}^{*}\leftarrow \pi^{\prime}$\tcp*{\footnotesize{the optimal spot price at slot $t$}}

\caption{\enskip SpotiPrice($M_{t}$, $\mathcal{A}_{t}$, $f_{t}$, $L$, $\beta$)}\label{optimal-pricing}
\end{algorithm}

At every slot $t$, the expression (\ref{optimal-pricing-decision}) could be used to decide the optimal spot price, and the corresponding procedure is presented in Algorithm~\ref{optimal-pricing}: it checks every possible value in $\mathcal{V}_{t}$ to see which can maximize the revenue function (\ref{equa-gain}). {\em A key feature of our algorithm} is that such decisions are implementable in practice, since the CSP has full knowledge of all parameters in $\mathcal{G}(t)$ except the control parameter, i.e., the spot price $\pi_{t}$, at every $t$.

\subsection{Running Spot and On-demand Services}
\label{sec.framework-on-demand-spot-services}

So far, we have shown in Sec.~\ref{sec.basic-resource-management} and Sec.~\ref{sec.optimal-spot-pricing} an integral framework for running spot and on-demand services. Now, we explain how this framework works as a whole.

\vspace{0.12em}\noindent\textbf{Sharing model.} The discrete-time service model is proposed in Section~\ref{sec.model_colocating} for sharing server time among on-demand and spot jobs where time is divided into consecutive slots. The jobs that arrive in the period of slot $t-1$ will be assigned at the next slot $t$, as illustrated in Fig.~\ref{Fig.discrete-time-model}, where $t=1, 2, \cdots$. The action (i.e., job assignment) that changes the states of servers happens only at the beginning of each slot $t$ and their states keep constant along the time slot. After the job assignment, an on-demand job $j$ will utilize the server for $s_{j}$ slots while an accepted spot job can stably access the server for one slot.

\vspace{0.12em}\noindent\textbf{Job assignment, pricing and acceptance.} While running on-demand and spot services, several actions are coordinated to control the job's access to servers and they occur sequentially at every slot $t=1, 2, 3, \cdots$. At $t$, on-demand jobs are first dispatched to servers and then two processes happen. The first is the calculation of the optimal spot price $\pi_{t}^{*}$ at $t$: with the framework in Section~\ref{sec.basic-resource-management} (illustrated in Fig.~\ref{Fig.framework-1}), we could determine by Algorithm~\ref{procedure-accepting-bids} the bids accepted at $t$ under an arbitrary spot price $\pi_{t}$. Based on this, by the framework illustrated in Fig.~\ref{Fig.framework-2}, $\pi_{t}^{*}$ could be derived. The second process actually determines the acceptance and assignment of spot jobs: $\pi_{t}^{*}$ is the final spot price announced to users and is actually used as the input of the framework illustrated in Fig.~\ref{Fig.framework-1}; then, the actions of accepting bids and assigning spot jobs to servers happen.

We have concluded the presentation of the basic framework. Finally, we analyze the spot user's behavior from a game theory perspective. For every spot user $j$, the willingness-to-pay (WTP) of $j$ is the maximum price at or below which it can accept the spot service and we denote its WTP by $c_{j}$. If the bid of $j$ is accepted and it gets spot service at $t$, let $\varsigma_{j}$ denote the effective server utilization time during the slot $t$; let $\hat{o}_{j}=(c_{j}-\pi_{t})\cdot \frac{\varsigma_{j}}{L}$, and the payoff of $j$ equals $\hat{o}_{j}$. If its bid is rejected, the payoff equals zero. In any case, we denote by $o_{j}$ the payoff of $j$. A spot user $j$ will truthfully report its WTP if its payoff $o_{j}$ is maximal or at least not less by being truthful, regardless of what the others do.
\begin{proposition}\label{proposi-truthfulness}
At every slot $t$, when a user $j$ bids a price for spot instances, it will truthfully report its WTP to the CSP, i.e., its bid price $v_{j}$ equals its WTP $c_{j}$.
\end{proposition}
\begin{proof}
There are $A_{t}$ bids whose prices satisfy (\ref{bid-order}). By Lemma~\ref{lemma-optimal-price} and Algorithm~\ref{procedure-accepting-bids}, we have in an optimal solution that the $N_{t}$ bids with the highest prices will be accepted and $\pi_{t}=v_{N_{t}}$. Thus, we have that (1) all bids whose prices are larger than $\pi_{t}$ will be accepted, (2) a bid whose price equals $\pi_{t}$ is possibly accepted, subject to the capacity constraint, and (3) all bids whose prices are smaller $\pi_{t}$ will be rejected. It suffices to show that the payoff of $j$ is maximal or not less by being truthful, regardless of the value of $\pi_{t}$. When $j$ misreports its WTP, there are two cases: (\rmnum{1}) $v_{j}>c_{j}$ and (\rmnum{2}) $v_{j}<c_{j}$. First, we analyze the first case. (\rmnum{1}.a) If $\pi_{t}> v_{j}$, user $j$ is rejected with $o_{j}=0$ no matter whether it is truthful. (\rmnum{1}.b) If $\pi_{t}= v_{j}$, it may be accepted or not; we have $o_{j}=\hat{o}_{j}<0$ if accepted and $o_{j}=0$ otherwise. By being truthful, it is rejected with $o_{j}=0$. (\rmnum{1}.c) If $c_{j}<\pi_{t}<v_{j}$, it is accepted  with $o_{j}=\hat{o}_{j}<0$; by being truthful, it is rejected with $o_{j}=0$. (\rmnum{1}.d) If $c_{j}=\pi_{t}$, it is accepted with $o_{j}=0$; by being truthful, we also have $o_{j}=0$ no matter whether it is accepted. (\rmnum{1}.e) If $\pi_{t}<c_{j}$, it is accepted and gets the same payoff $o_{j}$ no matter whether it is truthful. Next, we analyze the second case in a similar way. (\rmnum{2}.a) If $\pi_{t}< v_{j}$, it is accepted and gets the same payoff no matter whether it is truthful. (\rmnum{2}.b) If $\pi_{t}= v_{j}$, it may be accepted or not with the payoff $o_{j}$ equaling $\hat{o}_{j}$ or 0; by being truthful, it is accepted with $o_{j}=\hat{o}_{j}>0$. (\rmnum{2}.c) If $v_{j}<\pi_{t}< c_{j}$, it is rejected with $o_{j}=0$; by being truthful, it is accepted with $o_{j}=\hat{o}_{j}>0$. (\rmnum{2}.d) If $\pi_{t}= c_{j}$, it is rejected with $o_{j}=0$; by being truthful, it may be accepted or not with $o_{j}=0$ in any case. (\rmnum{2}.e) If $\pi_{t}> c_{j}$, it is rejected with $o_{j}=0$ no matter whether it is truthful. Observing in the analysis above, we can conclude that the payoff of $j$ is maximal or not less by being truthful; thus, the proposition holds.
\end{proof}


\section{An Extended Framework}
\label{sec.extended-model}

In this section, we provide further insight into the basic framework introduced Sec.~\ref{sec.basic-resource-management} and Sec.~\ref{sec.optimal-spot-pricing}, and propose an extension attaining higher quality of service.




\subsection{Limitation to Usability}
\label{sec.motivation-further-design}

In the basic framework, the slot duration has two implications as indicated in Feature~\ref{feature-11} and Feature~\ref{feature-12}. First, since an on-demand job may arrive at any time point in the period of a slot and will get served at the beginning of the next slot, the delivery of computing service is delayed to some extent. In the worst case up to $k$ minutes are required before being allocated. This can harm the quality of on-demand service and
\begin{itemize}
\item from an on-demand user's perspective, it may hope that the slot duration is not large.
\end{itemize}
Second, upon acceptance of a spot job at a slot $t$, a server will be allocated to this job. The slot duration represents the guaranteed time that the server will dedicate to the job. Such a spot job risks being rejected at the next slot $t+1$ since the spot price may change. Furthermore, spot jobs newly accepted or spot jobs migrated to another server to prioritize new on-demand jobs face a process of migrating or loading VMIs, which takes $k^{\prime}$ minutes, e.g., $k^{\prime}$ approximates 3 under current technology. Ultimately, such spot jobs can effectively utilize servers for $k-k^{\prime}$ minutes, where $k>k^{\prime}$. So, we have
\begin{itemize}
\item from a spot user's perspective, it may hope that the slot duration is properly large.
\end{itemize}
In the basic framework, both the persistence of spot service and the immediacy of on-demand service depend on the slot duration; it is difficult to simultaneously satisfy the requirements of on-demand and spot users. In fact, in Amazon EC2, $k$ is set to 5; in this case, upon acceptance of a bid at $t$, $3$ minutes are wasted while $2$ minutes are effectively utilised, as illustrated in Fig.~\ref{Fig.extended-model-motivation-4}; however, an on-demand user may need to wait for up to 5 minutes to get served. Finally, we observe that the process of loading or migrating VMIs requires additional system resources (e.g., bandwidth) to be consumed. From a system administrator's perspective, the convenience of the spot-pricing scheme may be reduced if it generates a large number of such operations.


\begin{figure}[t]
\begin{center}
  \includegraphics[width=2.85in]{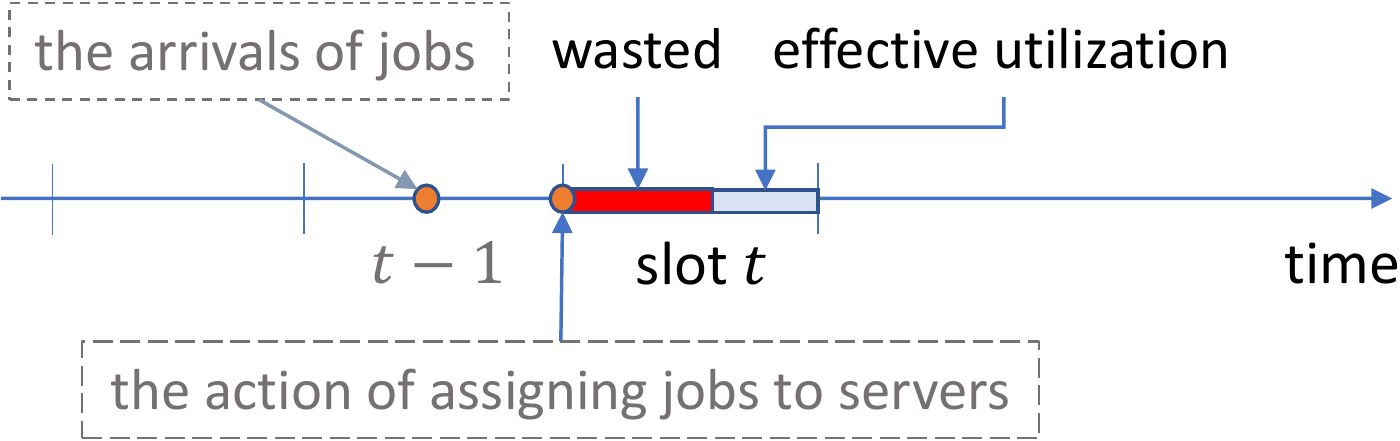}
  \caption{The limits of usability of the basic model: the length of the red rectangle is the time used for loading for loading/migrating VMIs, where no task can be executed.}\label{Fig.extended-model-motivation-4}
\end{center}
\end{figure}

\subsection{Our Improvement: Implementing Spot Pricing in Parallel}
\label{sec.spot-pricing-in-parallel}

In this subsection, we propose an extended framework to solve the tradeoff between the persistence of spot service and the immediacy of on-demand service; here, the basic framework of Sec.~\ref{sec.basic-resource-management} and Sec.~\ref{sec.optimal-spot-pricing} will be implemented in parallel on multiple groups of servers that alternatingly serve the jobs that arrive at different slots.

\vspace{0.7em}\noindent\textbf{Parallelized model.} There are a total of $m$ servers. In the basic framework, every job accepted at $t$ will be assigned to one of the $m$ servers where $t=1, 2, 3, \cdots$. Now, the servers are divided into $b$ groups, and the $i$-th group consists of $m_{i}$ servers where $\sum_{i=1}^{b}{m_{i}}=m$. For all on-demand and spot jobs accepted at any slot $t$, there exists a $i\in [1, b]$ such that all these jobs will be assigned to the servers of the $i$-th group; here, $t$ and $i$ satisfy the following relation:
\begin{align}\label{equa-relation}
h=\left\lceil \frac{t}{b} \right\rceil-1 \text{ and } i = t - h\cdot b,
\end{align}
where $b$ is a system parameter; for example, if $b=2$, the jobs arriving at slot $t=1, 3, 5, \cdots$ will be served by the first group of servers. In other words, for all $i\in [1, b]$, the $i$-th group is an independent processing unit that serves the on-demand and spot jobs that arrive at slot $t=h\cdot b+i$ (i.e., in the period of slot $t-1$) where $h=0, 1, 2, \cdots$; these jobs arrive every $b$ slots. 
Within any group, the schemes for processing jobs are similar to the ones in the basic framework, which will be elaborated later.

At the $i$-th group, the action of assigning jobs to servers only occurs at the beginning of slot $t=h\cdot b + i$ where $h = 0, 1, 2, \cdots$. The key observation is that only such actions will change the server states, and the server states of the $i$-th group keep constant for $b$ slots, i.e., $b\cdot k$ minutes. Now, we have two parameters $k$ and $b$ to control the jobs' waiting and service time. After the assignment of jobs at $t$, the spot jobs can utilize the assigned server for $b$ slots without being interrupted by the high priority of on-demand jobs that arrive at the subsequent $b-1$ slots; those latter on-demand jobs will be served by the other $b-1$ groups of servers. Such an extended framework is also illustrated in Fig.~\ref{Fig.extended-model-motivation-2} where $b=3$, and it has the following two features that quantify the immediacy of on-demand service and the persistence of spot service:

\begin{feature}\label{feature-21}
Upon arrival, on-demand users need to wait for up to $k$ minutes to get served.
\end{feature}

\begin{feature}\label{feature-22}
Once spot users bid successfully, it is guaranteed that they can persistently get served for $b\cdot k$ minutes.
\end{feature}

\begin{figure}[t]
\begin{center}
  \includegraphics[width=3.95in]{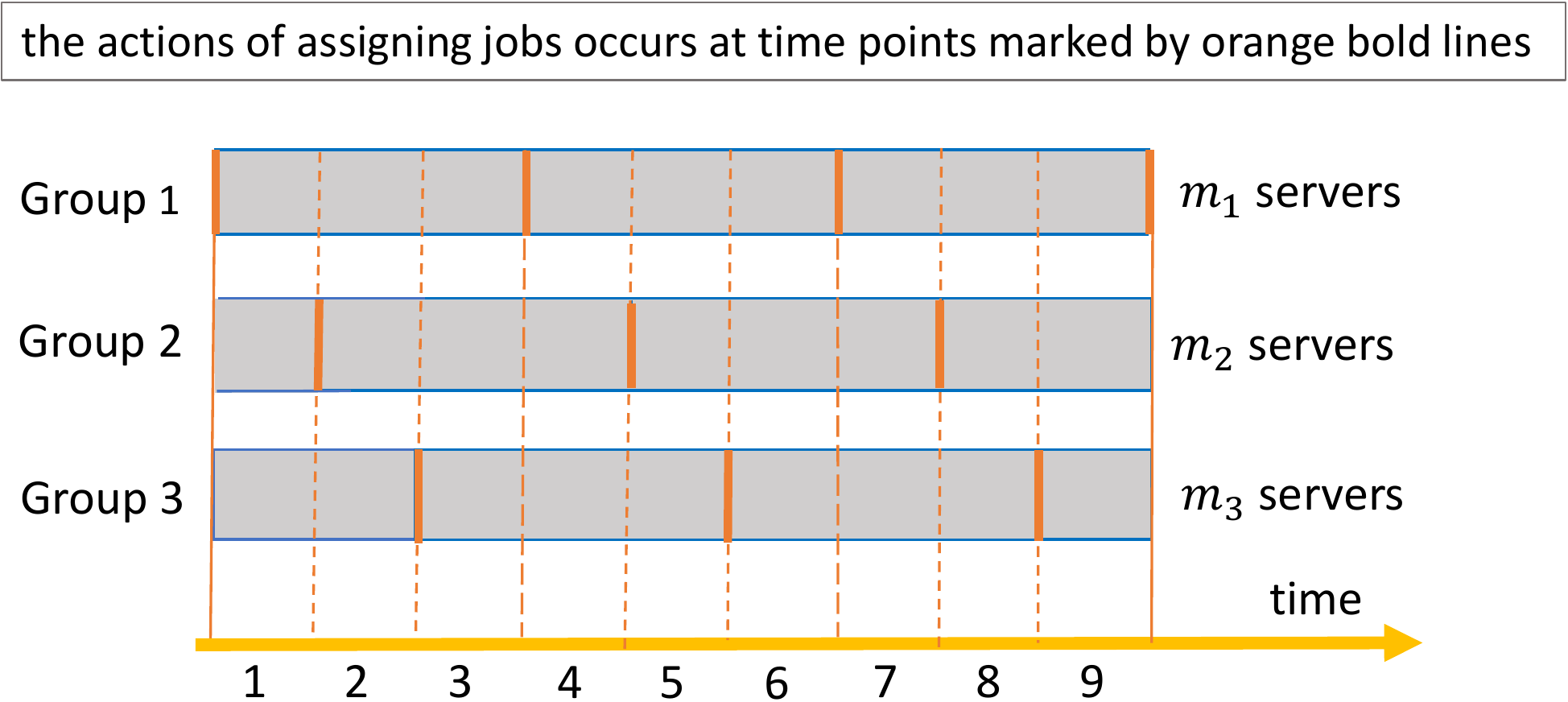}
  \caption{An extended framework: for all $i\in [1, 3]$, the $i$-th group serves the jobs arriving and accepted at slot $3\cdot h+i$ where $h=0, 1, 2, \cdots$.}\label{Fig.extended-model-motivation-2}
\end{center}
\end{figure}

\noindent\textbf{Parameter Setting.} In this framework, the parameters $k^{\prime}$, $b$, $k$, and $L$ are set to satisfy the following relations: (\rmnum{1}) $k^{\prime} < b\cdot k$, and (\rmnum{2}) $L$ is the multiple of $b$, i.e., $K=\frac{L}{b}$ is an integer.
The first relation guarantees for each accepted bid that the effective server utilization time is positive, i.e., $b\cdot k-k^{\prime}>0$; here, the value of $k$ can be smaller than $k^{\prime}$. Instead, to guarantee this, it is required that $k>k^{\prime}$ in the basic framework of Sec. \ref{sec.basic-resource-management} and Sec. \ref{sec.optimal-spot-pricing}.
Thus, in the extended framework, we can set the length of a slot to a small value, which can well guarantee the immediacy of on-demand service by Feature~\ref{feature-21}. For example, we set $k=1$ (minute); then, an on-demand job needs to wait for at most 1 minute to get served. On the other hand, we can set $b$ to a large value such that $b\cdot k>k^{\prime}$, which can well guarantee the persistence of spot service by Feature~\ref{feature-22}. For example, when $b=30$, the minimum time that a server dedicates to an accepted bid is up to 30 minutes. As a result, with the extended framework, we can well address the two seemingly conflicting requirements from spot and on-demand users above. We also illustrate in Fig.~\ref{Fig.extended-model-motivation-5} another case where $k=5$, $b=3$, and $k^{\prime}=3$, in contrast to Fig.~\ref{Fig.extended-model-motivation-4}; here, upon acceptance of a bid, 3 minutes are wasted and 12 minutes are effectively utilized.

\begin{figure}[t]
\begin{center}
  \includegraphics[width=2.95in]{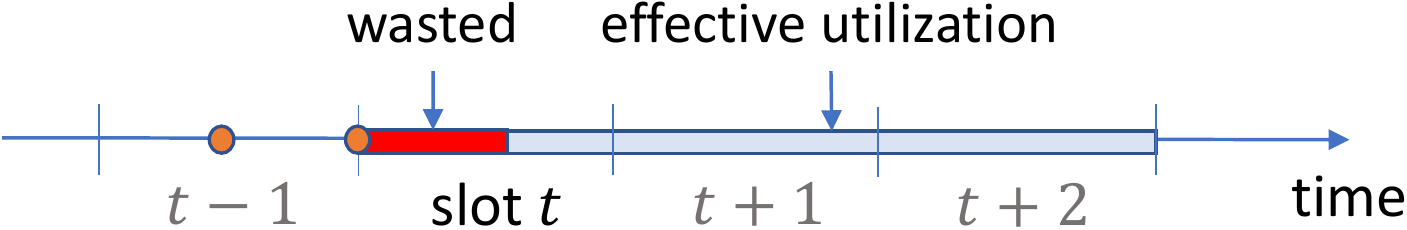}
  \caption{Improvement to Usability.}\label{Fig.extended-model-motivation-5}
\end{center}
\end{figure}

At the $i$-th group, on-demand jobs arrive at $t=h\cdot b+i$ where $h=0, 1, 2, \cdots$; they have higher priority to access servers and are first dispatched to the servers.
Recall that the reference time interval for billing an on-demand user is made of $L$ consecutive slots. Thus, the size $s_{j}$ of an on-demand job $j$ can be viewed as the multiple of $L$ (slots) and let $\tau_{j}=\frac{s_{j}}{L}$. After the assignment of $j$ at slot $t$, it will occupy a server for $s_{j}=b\cdot K\cdot \tau_{j}$ slots. The second relation guarantees that the number of slots occupied by $j$ is the multiple of $b$. Thus, for all $t^{\prime}=h^{\prime}\cdot b+i$ where $h^{\prime}=0, 1, 2, \cdots$, if the slot $t^{\prime}$ of a server is occupied by an on-demand job, the entire period of $[t^{\prime}, t^{\prime}+b-1]$ will be occupied by this job.


\vspace{0.35em}\noindent\textbf{Job assignment, acceptance and pricing.} For all $i\in [1, b]$, the $i$-th processing unit is a single system where the way of processing jobs is the same as the way of the basic framework, except that there are only $m_{i}$ servers available and jobs arrive every $b$ slots, i.e., at $t=h\cdot b+i$ where $h=0, 1, 2, \cdots$. In the following, we apply the notion in the basic framework to the scenario here and show the process of assigning and accepting jobs, as described in Section~\ref{sec.on-demand-jobs}, \ref{sec.pricing} and \ref{sec.dispatch} and illustrated in Fig.~\ref{Fig.framework-1}; then, we explain the process of determining the optimal spot price at $t$, as described in Section~\ref{sec.optimal-spot-pricing} and illustrated in Fig.~\ref{Fig.framework-2}.

After the assignment of on-demand jobs, the idle servers of the $i$-th unit in the period of $[t, t+b-1]$ are sold as spot instances where $t=h\cdot b+i$ and $h=0, 1, 2, \cdots$. We denote their amount by $M_{t}^{(i)}$, which is also the capacity of accepting bids at $t$. Spot users arrive at every such slot $t$ and bid prices to utilize spot instances.  All bids available at $t$ are denoted by $\mathcal{A}_{t}=\mathcal{J}_{i, t}^{\prime}\cup\mathcal{J}_{i, t}^{\prime\prime}$: $\mathcal{J}_{i, t}^{\prime}$ denotes the bids that belongs to the spot users whose bids have been accepted at $t-b$ and who continue bidding at $t$, and $\mathcal{J}_{i, t}^{\prime\prime}$ denotes the bids newly arriving and submitted at any time point in the period of slot $t-1$. Replacing the $\mathcal{J}_{t}^{\prime}$, $\mathcal{J}_{t}^{\prime\prime}$ and $M_{t}$ of Algorithm~\ref{procedure-accepting-bids} with $\mathcal{J}_{i, t}^{\prime}$, $\mathcal{J}_{i, t}^{\prime\prime}$ and $M_{t}^{(i)}$, we could get the procedure for accepting bids in this section.
As is given in (\ref{equa-accepted-bids}), the number of accepted bids at $t$ depends on the spot price $\pi_{t}$, the available bids $\mathcal{A}_{t}$ and the number of spot instance $M_{t}^{(i)}$ and is
\begin{align}\label{equa-accepted-bids-2}
N_{t}^{(i)} = \min\left\{M_{t}^{(i)},\, |F(\pi_{t},\, \mathcal{A}_{t})|\right\}.
\end{align}
Let $\hat{\mathcal{J}}_{i, t}^{\prime}$ denote the accepted bids in $\mathcal{J}_{i, t}^{\prime}$, and $\hat{\mathcal{J}}_{i, t}^{\prime\prime}$ denote the accepted bids in $\mathcal{J}_{i, t}^{\prime\prime}$. The procedure of assigning spot jobs to servers is the same as the procedure in Section~\ref{sec.dispatch} after (\rmnum{1}) we replace the accepted bids $\hat{\mathcal{J}}_{t}^{\prime}$ and $\hat{\mathcal{J}}_{t}^{\prime\prime}$ with $\hat{\mathcal{J}}_{i,t}^{\prime}$ and $\hat{\mathcal{J}}_{i,t}^{\prime\prime}$, and (\rmnum{2}) we replace the period of slot $t-1$ in the Step 2 with the period of $[t-b, t-1]$. Once an instance is offered to a spot job, the instance will dedicate $b$ time slots to this job.

For the accepted bids where the operation of loading or migrating VMIs is needed, their amount is denoted by $f_{t}^{(i)}$, which is observable. Recall that $\beta=\frac{k^{\prime}}{k}$, and there are $f_{t}^{(i)}$ accepted bids whose effective server utilization time is $b\cdot k-k^{\prime}$ minutes, i.e., $b-\beta$ slots; for the other accepted bids, the effective utilization time is $b$ slots and their amount is $N_{t}^{(i)}-f_{t}^{(i)}$. As stated in Definition~\ref{def-billing-rule}, the instances will be charged for the period in which they are effectively utilized for executing workload, excluding the period in which VMIs are loaded or migrated; the price of effectively utilizing an instance for a slot is $\frac{\pi_{t}}{L}$. Similar to (\ref{equa-gain}), the revenue from the spot market at $t$ is as follows:
\begin{equation}\label{equa-gain-1}
\mathcal{G}^{(i)}(t)  = \left(N_{t}^{(i)}-f_{t}^{(i)}\right)\cdot \frac{\pi_{t}}{K} + f_{t}^{(i)}\cdot \left(1-\frac{\beta}{b}\right)\cdot \frac{\pi_{t}}{K}
                      = \left( N_{t}^{(i)}  - \frac{\beta}{b}\cdot f_{t}^{(i)}  \right)\cdot \frac{\pi_{t}}{K},
\end{equation}
where $K=\frac{L}{b}$. With (\ref{equa-accepted-bids-2}), $\mathcal{G}^{(i)}(t)$ can be transformed as a function of the single variable $\pi_{t}$:
\begin{align}\label{equa-gain-2}
\mathcal{G}^{(i)}(t) = \check{\mathcal{G}}\left(\pi_{t},\, M_{t}^{(i)},\, \mathcal{A}_{t},\, f_{t}^{(i)}\right),
\end{align}
where $M_{t}^{(i)}$, $\mathcal{A}_{t}$ and $f_{t}^{(i)}$ are observable at slot $t$. As we conclude in Lemma~\ref{lemma-optimal-price} and Proposition~\ref{proposi-optimal-price}, the optimal spot price at $t$ is in $\mathcal{V}_{t}$ and is such that
\begin{equation}\label{optimal-pricing-decision-1}
\pi_{t}^{*}\leftarrow \arg\max\limits_{\pi_{t}\in\mathcal{V}_{t}}{\check{\mathcal{G}}\left(\pi_{t},\, M_{t}^{(i)},\, \mathcal{A}_{t},\, f_{t}^{(i)}\right)};
\end{equation}
the corresponding procedure is presented by Algorithm~\ref{optimal-pricing}, i.e., SpotiPrice$\left(M_{t}^{(i)}, \mathcal{A}_{t}, f_{t}^{(i)}, K, \beta/b\right)$.

\vspace{0.12em}\noindent\textbf{Reference Performance Metric.} After giving the framework for running on-demand and spot services, a further objective of this paper is evaluating its performance. In the extended framework, the states of all servers still keep constant in the period of every slot; for all $l\in[1, b]$, we denote by $\overline{M}_{t}^{(l)}$ the number of servers of the $l$-th processing unit used as on-demand instances at slot $t$. Let $\overline{M}_{t}=\sum_{l=1}^{b}{\overline{M}_{t}^{(l)}}$, denoting the total number of instances that are occupied by on-demand jobs at $t$. Recall that the price of utilizing an on-demand instance for $L$ slots is $p$, and at slot $t$ the revenue from the on-demand market is
\begin{align}\label{equa-on-demand-revenue-1}
\mathcal{G}_{t}^{o}=\frac{p}{L}\cdot\overline{M}_{t}.
\end{align}
In contrast, the revenue from the spot market at $t$ is $\mathcal{G}^{(i)}(t)$ given in (\ref{equa-gain-1}). The revenue improvement brought by the spot market is measured by the following ratio:
\begin{equation}\label{extended-efficiency}
\alpha_{t} = \frac{\mathcal{G}^{(i)}(t)}{\mathcal{G}_{t}^{o}}.
\end{equation}
The ratio $\alpha_{t}$ represents how much the CSP's revenue could be improved by at slot $t$ after complementing the on-demand market with a spot market. The revenue improvement $\alpha_{t}$ is a main performance metric of this paper.

\subsection{An Efficiency Analysis}
\label{sec.efficiency-analysis}

The extended framework of this section allows us to use some assumptions to derive an analytical result. The aim is to characterize the revenue improvement $\alpha_{t}$ with a simple mathematical expression; it helps us clearly understand which factors are affecting the revenue improvement when the on-demand market is complemented with a spot market via the framework proposed in this paper. The two assumptions are: (\rmnum{1}) the product $b\cdot k$ is set to a large enough value, and (\rmnum{2}) at every slot $t$, the prices of all bids $\mathcal{A}_{t}$ follow a uniform distribution over $[\underline{\pi}, \overline{\pi}]$, as is used for cloud services in \cite{Zheng15}; here, $\overline{\pi}$ and $\underline{\pi}$ are the maximum and minimum bid prices of users, and $\overline{\pi}$ can be viewed as the on-demand price $p$, i.e., $\overline{\pi}=p$. We emphasize that our framework itself do not rely on such assumptions in order to run on-demand and spot services.

Under the first assumption we have $\frac{k^{\prime}}{b\cdot k}=\frac{\beta}{b} \rightarrow 0$ where $k^{\prime}$ is a fixed parameter; furthermore, $f_{t}^{(i)}/N_{t}^{(i)}\leq 1$. Thus, by (\ref{equa-gain-1}), the revenue from the spot market at every slot $t$ can be approximated as
\begin{align}\label{equa-spot-revenue-1}
\mathcal{G}^{(i)}(t) = N_{t}^{(i)}\cdot \frac{\pi_{t}}{K}.
\end{align}
Under the second assumption we have that the expected number of bids whose prices are not below $\pi_{t}$ is $\left|F(\pi_{t},\, \mathcal{A}_{t})\right|=A_{t}\cdot (\overline{\pi}-\pi_{t})/(\overline{\pi}-\underline{\pi})$. The number of accepted bids given in (\ref{equa-accepted-bids-2}) can be transformed as
\begin{equation}\label{equa-accepted-bids-1}
N_{t}^{(i)} = \left|F(\pi_{t},\, \mathcal{A}_{t})\right| = A_{t}\cdot \frac{\overline{\pi}-\pi_{t}}{\overline{\pi}-\underline{\pi}},
\end{equation}
subject to the constraint that $N_{t}^{(i)}\leq M_{t}^{(i)}$. Let
\begin{align}\label{equa-pi}
\pi_{t}^{\prime} = \overline{\pi} - \frac{M_{t}^{(i)}}{A_{t}}\cdot\left(\overline{\pi}-\underline{\pi}\right)\text{, and  }\pi_{t}^{\prime\prime} =\max\left\{\pi_{t}^{\prime}, \underline{\pi}\right\}.
\end{align}
With (\ref{equa-accepted-bids-1}), the constraint translates to $\pi_{t}\geq \pi_{t}^{\prime}$. Since $\pi_{t}\in [\underline{\pi}, \overline{\pi}]$, the spot price at $t$ should satisfy
\begin{align}\label{equa-accepted-bids-constraints}
\pi_{t}\geq \pi_{t}^{\prime\prime}.
\end{align}
Finally, at any slot $t=h\cdot b+i$, we have from (\ref{equa-spot-revenue-1}) and (\ref{equa-accepted-bids-1}) that, the CSP's revenue from the spot market is as follows
\begin{align}\label{equa-gain-3}
\mathcal{G}^{(i)}(t) = \frac{A_{t}}{K\cdot \left(\overline{\pi}-\underline{\pi}\right)}  \cdot \left(\overline{\pi}\cdot\pi_{t}-\pi_{t}^{2}\right).
\end{align}
At $t$, $A_{t}$ is observable; $\mathcal{G}^{(i)}(t)$ is a quadratic function of $\pi_{t}$ subject to (\ref{equa-accepted-bids-constraints}). We let $\rho=\underline{\pi}/\overline{\pi}$ and have that
\begin{proposition}\label{theorem-performance-analysis}
At every slot $t$, the optimal spot price $\pi_{t}^{*}$ and the maximum revenue from the spot market are as follows:
  \begin{align}\label{extended-gain-2}
  \mathcal{G}^{(i)}(t) =
\begin{dcases}
\frac{\overline{\pi}}{4\cdot K}\cdot\frac{A_{t}}{1-\rho} &  \text{if } \rho\leq \min\left\{0.5, 1-D/2\right\}, \text{where } \pi_{t}^{*}=\overline{\pi}/2\\
 A_{t}\cdot \frac{\underline{\pi}}{K}  &  \text{if } D \leq 1 \text{ and } \rho > 0.5, \text{where } \pi_{t}^{*}=\underline{\pi}\\
 \frac{\overline{\pi}}{K}\cdot ( 1-\frac{1-\rho}{D} )\cdot M_{t}^{(i)}  &  \text{if } D > 1 \text{ and } \rho > 1-D/2, \text{where } \pi_{t}^{*}=\pi_{t}^{\prime}.
\end{dcases}
  \end{align}
where $\pi_{t}^{\prime}$ is given in (\ref{equa-pi}), $D=\frac{A_{t}}{M_{t}^{(i)}}$, and $\rho\in (0,1)$.
\end{proposition}
\begin{proof}
For the quadratic function $\mathcal{G}^{(i)}(t)$, the axis of symmetry is a vertical line $x=\frac{\overline{\pi}}{2}$; its maximum value is achieved at (\rmnum{1}) $\pi_{t}=\frac{\overline{\pi}}{2}$ if $\frac{\overline{\pi}}{2} \in [\pi_{t}^{\prime\prime}, \overline{\pi}]$, and at (\rmnum{2}) $\pi_{t}=\pi_{t}^{\prime\prime}$ if $\frac{\overline{\pi}}{2} < \pi_{t}^{\prime\prime}$. In each case, such $\pi_{t}$ is the optimal spot price. In the latter case, there are two subcases: (\rmnum{2}.a) if $\underline{\pi} \geq \pi_{t}^{\prime}$, the optimal spot price $\pi_{t}^{*}$ is $\underline{\pi}$; (\rmnum{2}.b) if $\underline{\pi} < \pi_{t}^{\prime}$, $\pi_{t}^{*}=\pi_{t}^{\prime}$. In the case (\rmnum{1}), the condition $\frac{\overline{\pi}}{2} \in [\pi_{t}^{\prime\prime}, \overline{\pi}]$ is equivalent to the condition $\frac{\overline{\pi}}{2} \geq \pi_{t}^{\prime\prime}$, which requires both $\frac{\overline{\pi}}{2} \geq \pi_{t}^{\prime}$ and $\frac{\overline{\pi}}{2} \geq \underline{\pi}$ by (\ref{equa-pi}); due to $\rho=\underline{\pi}/\overline{\pi}$, this condition is equivalent to the condition $\rho\leq \min\{0.5,\, 1-D/2\}$. In the case (\rmnum{2}), the condition is equivalent to $\rho > \min\{0.5,\, 1-D/2\}$. The condition $\underline{\pi} \geq \pi_{t}^{\prime}$ in the case (\rmnum{2}.a) is equivalent to $D\leq 1$; thus, the conditions to make $\pi_{t}^{*}=\underline{\pi}$ are $\rho > \min\{0.5,\, 1-D/2\}$ and $D\leq 1$, which are further equivalent to $\rho > \frac{1}{2}$ and $D\leq 1$. Similarly, in the case (\rmnum{2}.b), the conditions to make $\pi_{t}^{*}=\pi_{t}^{\prime}$ are $\rho > 1-D/2$ and $D > 1$. Finally, substituting the optimal spot price $\pi_{t}^{*}$ in each case into (\ref{equa-gain-3}), we could get the maximum revenue $\mathcal{G}^{(i)}(t)$ in (\ref{extended-gain-2}).
\end{proof}

Finally, we can quantify the revenue improvement $\alpha_{t}$ and the conclusion below follows directly from (\ref{extended-efficiency}) and Proposition~\ref{theorem-performance-analysis}.

\begin{corollary}\label{corollary-improvement}
  The revenue improvement brought by the spot market is as follows:
  \begin{align}\label{extended-efficiency-1}
  \alpha_{t} =
\begin{dcases}
\frac{1}{4}\cdot \frac{1}{1-\rho}\cdot D\cdot I  &  \text{if } \rho\leq \min\left\{0.5, 1-D/2\right\},  \text{where } \pi_{t}^{*}=\overline{\pi}/2\\
\rho \cdot D\cdot  I  &  \text{if } D \leq 1 \text{ and } \rho > 0.5, \text{where } \pi_{t}^{*}=\underline{\pi}\\
 ( 1-\frac{1-\rho}{D} )\cdot I  &  \text{if } D > 1 \text{ and } \rho > 1-D/2, \text{where } \pi_{t}^{*}=\pi_{t}^{\prime}
\end{dcases}
  \end{align}
  where $\pi_{t}^{\prime}$ is given in (\ref{equa-pi}), $D=\frac{A_{t}}{M_{t}^{(i)}}$, $I=\frac{M_{t}^{(i)}}{\overline{M}_{t}/b}$, and $\rho=\underline{\pi}/\overline{\pi}\in (0, 1)$.
\end{corollary}

Now, we explain the physical meaning of Corollary~\ref{corollary-improvement}. Recall that $M_{t}^{(i)}$ is the number of spot instances available to serve the bids at $t$, and they are idle instances in the on-demand market; $M_{t}^{(i)}$ determines the capacity of accepting bids. $A_{t}$ is the total number of bids at $t$. $\overline{M}_{t}$ is the number of instances executing on-demand jobs at $t$. The ratio $D=A_{t}/M_{t}^{(i)}$ can be viewed as {\em the saturation degree} of spot market at $t$, e.g., when it is larger than 1, the spot market is fully saturated with bids and not all bids could be accepted with the capacity constraint; when the ratio is zero, there are no bids at $t$. The ratio $\rho$ is {\em the value density} of user's bids. If $\rho$ is small, the price difference of users' bids is large. Let us consider a scenario where the maximum bid price of users $\overline{\pi}$ is given and a fixed number of bids with the highest prices are accepted: if $\rho$ is small, the lowest price of the accepted bids would be small; as a result, the spot price at $t$ is also small, as well as the revenue that the CSP gains from the bids. Similarly, we can have opposite conclusions for the case of a large $\rho$. In the long run, the mean of $\overline{M}_{t}/b$ approximates the mean of $\overline{M}_{t}^{(i)}$, and the ratio $I=\frac{M_{t}^{(i)}}{\overline{M}_{t}/b}$ could be roughly viewed as {\em the vacancy-to-utilization ratio} of on-demand market, representing the percentage of servers in idle states at slot $t$. If the vacancy-to-utilization ratio is large and the CSP does not offer spot service, only a small part of on-demand instances are effectively utilized by on-demand jobs and most of them are in idle states at $t$; for example, if the ratio is 7, it implies that 87.5\% of the instances will be in idle states.

By (\ref{extended-efficiency-1}) and (\ref{extended-efficiency}), the vacancy-to-utilization ratio of on-demand market $I$, the saturation degree of spot market $D$, and the value density $\rho$ together determine the revenue improvement $\alpha_{t}$ after complementing the on-demand market with a spot market. We have by (\ref{extended-efficiency-1}) that, the larger the value density $\rho$, the larger the revenue improvement. For example, in the case that $\rho$ is large and the saturation degree of spot market is low (i.e., $D \leq 1$ and $\rho > 0.5$), the best strategy is accepting all bids $\mathcal{A}_{t}$ at $t$, achieving the maximum $\alpha_{t}$. In the following, we consider another setting where the value density is small (i.e., $\rho\leq 0.5$). Suppose that the value density $\rho$ is 0.2. The vacancy-to-utilization ratio of on-demand market is fixed and is mainly determined by the QoS guarantee offered to the arriving on-demand jobs; since the jobs require a quick response from the CSP, its value is usually small, and we set $I$ to 7. Then, we have
  \begin{align}\label{extended-efficiency-example}
  \alpha_{t} =
\begin{dcases}
 2.1875\cdot D = 2.1875\cdot \left(A_{t}/M_{t}^{(i)}\right) & \text{ if } D \leq 1.6, \\
 7\cdot\left( 1-0.8\cdot\frac{1}{D} \right) = 7\cdot\left( 1-0.8\cdot \left(M_{t}^{(i)}/A_{t}\right) \right) & \text{ if }   D > 1.6.
\end{dcases}
  \end{align}
Here, if the saturation degree of spot market is large (i.e., $D=A_{t}/M_{t}^{(i)} > 1.6$), we have that after complementing the on-demand market with a spot market, the revenue improvement $\alpha_{t}$ is no smaller than $3.5$, representing at least 3.5-fold increase in the CSP's revenue. If the saturation degree is small (i.e., $A_{t}/M_{t}^{(i)} \leq 1.6$), the revenue improvement $\alpha_{t}$ is $2.1875$ times the saturation degree of spot market.

Finally, we observe that our derivation and observation above can be in principle generalised by relaxing assumption (\rmnum{2}) and letting $H(\pi_t)$ be a general probability distribution for the bid prices at every slot $t$ over the support $[\underline{\pi}, \overline{\pi}]$; then, the Equation (\ref{equa-accepted-bids-1}) becomes $N_{t}^{(i)} = A_{t}\cdot (1-H(\pi_{t}))$ and the revenue of spot market in (\ref{equa-gain-3}) becomes $G^{(i)}(t)=A_{t}\cdot (1-H(\pi_{t}))\cdot \pi_{t}/K$; the maximum revenue might be achieved when the spot price $\pi_{t}$ is such that the differential of $G^{(i)}(t)$ equals zero.
Under any distribution, one may observe that the revenue improvement $\alpha_{t}$ may mainly depend on the vacancy-to-utilization ratio of on-demand market $I$, the saturation degree of spot market $D$, and the value density $\rho$. For example, given the amount of servers, a larger $I$ means that more servers are idle for accepting bids; given the amount of spot instances and the distribution of bid prices, a larger $D$ means that more bids are available and the CSP can choose to accept the bids with higher prices; in both these cases, a higher revenue improvement may be achieved.



\section{Performance Evaluation}
\label{sec.performance-evaluation}

In this section we provide numerical validation for the proposed framework.

\subsection{Experimental Setting}
\label{sec.basic-experimental-setting}

Time is divided into consecutive slots and each slot contains $k=5$ minutes. On-demand instances are charged on an hourly basis and an hour contains $L=12$ slots. We set $b$ to 6, i.e., all servers are divided into 6 groups (also called processing units); for all $i\in [1, 6]$, the $i$-th group has $m_{i}$ servers and its value will be given in Section~\ref{sec.on-demand-instances-idle-empirical}. For all $i\in [1, 6]$, the $i$-th group is used to process the on-demand and spot jobs that arrive at slots $t=6\cdot h+i$ where $h=0, 1, 2, \cdots$.

\subsubsection{Components of On-demand and Spot Services}
\label{sec.components}

We have explained in Section~\ref{sec.framework-on-demand-spot-services} the basic framework; the final framework divides all servers into 6 groups.
For all $i\in [1, 6]$, at the beginning of slot $t=6\cdot h+i$ where $h=0, 1, 2, \cdots$, the job pricing, acceptance, and assignment of the $i$-th group are similar to the ones of the basic framework, and there are four main components:
\begin{description}

\item [1. Dispatching on-demand jobs.] We use the PTC policy to dispatch jobs to servers, as described in Section~\ref{sec.on-demand-jobs}.  On-demand users can be latency-critical and user-facing services have to meet strict tail-latency requirements at the 99th percentile of the distribution \cite{Delimitrou14a}; in other words, for every 100 jobs, there is at most one job that will miss its deadline. For the $i$-th processing unit, there are $m_{i}$ servers where $m_{i}$ is the minimum number of servers needed to guarantee the latency requirement and its value is given in Section~\ref{sec.on-demand-instances-idle-empirical}; if the number of servers is larger than $m_{i}$, more servers will be idle in the on-demand market and a higher revenue from the spot market may be achieved.

\item [2. Accepting spot jobs.] As explained in Section~\ref{sec.spot-pricing-in-parallel}, we apply Algorithm~\ref{procedure-accepting-bids} to the $i$-th group for determining which bids are accepted at slot $t$.

\item [3. Assigning spot jobs to servers.] As explained in Section~\ref{sec.spot-pricing-in-parallel}, we apply the procedure proposed in Section~\ref{sec.dispatch} to the $i$-th group for assigning spot jobs at $t$.

\item [4. Optimally pricing spot instances.] We use the algorithm SpotiPrice$\left(M_{t}^{(i)}, \mathcal{A}_{t}, f_{t}^{(i)}, K, \beta/b\right)$, described by Algorithm~\ref{optimal-pricing}, to determine the spot price $\pi_{t}$ at $t$. The time spent on loading or migrating VMIs is set to 3 minutes, i.e., $k^{\prime}=3$; here $\beta=\frac{k^{\prime}}{k}=0.6$. Finally, we use (\ref{equa-gain-1}) to determine the revenue from the spot market at $t$.
\end{description}

\subsubsection{Performance Metrics}
\label{sec.performance-metrics}

Recall that $\alpha_{t}$ is defined in (\ref{extended-efficiency}) and it is the ratio of the revenue of spot market to the revenue of on-demand market at a slot $t$. The {\em average revenue improvement} per slot is the average value of all $\alpha_{t}$ where $t=1, 2, 3, \cdots$ and it is denoted by $\alpha_{e}$. In our experiments, the main performance metric is $\alpha_{e}$ and it represents how much the CSP's revenue is increased by after complementing the on-demand market with a spot market.

Furthermore, for all $i\in [1, 6]$, the $i$-th processing unit is used to process the jobs arriving at slot $t=6\cdot h+i$ where $h=0, 1, 2, \cdots$. We assume that a super-slot contains 6 slots (i.e., 30 minutes). From the $i$-th slot on, the server state of the $i$-th unit changes every super-slot; the $h^{\prime}$-th super-slot corresponds to the period of $[t, t+5]$ where $h^{\prime}=1, 2, \cdots$ and $t=6\cdot (h^{\prime}-1)+i$. In our experiments, we also show the {\em average server utilization} per super-slot. In the case that the CSP offers both on-demand and spot instances (resp. on-demand instances alone), the utilization at a specific super-slot $h^{\prime}$ is defined as the ratio of the number of instances occupied by on-demand and spot jobs (resp. on-demand jobs) to the total number of instances available (i.e., $m_{i}$), denoted by $\theta_{t}^{(i)}$, where $h^{\prime}=(t-i)/6+1$.
In both cases, the average server utilization per super-slot is simply defined as the average value of $\theta_{1}^{(1)}, \theta_{2}^{(2)}, \cdots, \theta_{6}^{(6)}, \theta_{7}^{(1)}, \cdots, \theta_{12}^{(6)}, \theta_{13}^{(1)}, \cdots$, which is denoted by $\theta$ when only on-demand instances are offered and by $\mu$ when both instances are offered.
With the values of $\theta$ and $\mu$, we can numerically see the improvement to server utilization after complementing the on-demand market with a spot market.

As analyzed in Section~\ref{sec.efficiency-analysis}, the revenue improvement may mainly depend on three factors: (\rmnum{1}) the vacancy-to-utilization ratio of on-demand market, (\rmnum{2}) the saturation degree of spot market, and (\rmnum{3}) the distribution of users' bid prices. The first factor is mainly determined by the QoS guarantee offered to on-demand users and it is specified in Section~\ref{sec.components}; the related results are given in Section~\ref{sec.on-demand-instances-idle-empirical}. So, the environments of our experiments will vary in terms of the distribution of users' bid prices, and the arrival rate of bids; the main results will be given in Section~\ref{sec.basic-spot-pricing}.



\subsubsection{Arrival of on-demand jobs}
\label{sec.job-arrivals}

A public CSP such as Amazon EC2 serves many users of different sources; it is representative to use a heavy-tailed distribution to model the job size and a poisson distribution to model the job arrival \cite{Zhang18a,Chen11a}, which has been validated by some measurement study \cite{Zheng16a}. At every slot $t=1, 2, 3, \cdots$, the number of job arrivals follows a poisson distribution with a mean $\lambda_{o}$. A job's size is a random variable $x$ that follows a bounded pareto distribution with a scale parameter $x_{m}$ and a shape parameter $\alpha$; $x$ ranges in $[x_{m}, \overline{x}]$. Since on-demand instances are charged on an hourly basis, we further set the sizes of the jobs submitted to the CSP to $12\cdot\lceil \frac{x}{12} \rceil$. For a job $j$ with size $s_{j}$ and arrival time $a_{j}$, its deadline is $d_{j}=a_{j}+s_{j}-1$. The on-demand price is normalized as 1. In the whole system, there are 6 processing units and for all $i\in [1, 6]$, the $i$-th unit can be viewed as a single system used to process the jobs arriving at slot $t=6\cdot h+i$ where $h=0, 1, 2, \cdots$. Thus, each unit will process the jobs with the same statistical feature in terms of the arrival rate and the job size.
In the experiments, $\lambda_{o}$, $x_{m}$ and $\alpha$ are set to 60, 6, and $\frac{7}{6}$ respectively; then, the mean of that pareto distribution is 42 (slots); the lower and upper bounds of the job size are 0.5 and 13 hours. An exception occurs in Sec.~\ref{sec.on-demand-instances-idle-empirical} where we take different values for $\lambda_{o}$, $\alpha$, $\overline{x}$ to show their effect on the idleness of on-demand market.


\subsubsection{Arrival and departure of spot jobs}
\label{sec.job-arrivals-1}

At every slot $t$, there are some users that newly arrive and bid prices for spot instances; among these users, we assume that their bid prices follow some probability distribution.
At every $t$, the bids of these users vary in terms of the arrival rate, and their value distribution. In our simulations we consider two value distributions. The first is a uniform distribution over $[0.2, 1]$, following \cite{Zheng15}. The other is a bounded pareto distribution with a pareto index $\alpha=2$ and a scale parameter $x_{m}=0.3$, following \cite{Ben-Yehuda13a}; its lower and upper bounds are 0.3 and 1. Given a pareto distribution with $\alpha=2$ and $x_{m}=0.3$, its mean is $\frac{\alpha\cdot x_{m}}{\alpha-1}=0.6$, and the probability that a random variable $x$ takes on a value larger than $x_{0}$ is $(\frac{x_{m}}{x_{0}})^{\alpha}$, e.g., the probability that $x>0.6$ is 0.25. Thus, with the bounded pareto distribution, the proportion of the bids whose prices are small are larger.

All servers are divided into 6 groups; for all $i\in [1, 6]$, the $i$-th group serves the bids accepted at $t=6\cdot h+i$ where $h=0,1,2,\cdots$. The number $x_{t}^{(i)}$ of the bids that newly arrive at $t=6\cdot h+i$ follows a geometric distribution, similar to \cite{Zheng15}. Thus, $x_{t}^{(i)}$ is a random variable that denotes the number of failures before one success in a series of independent trials, where each trial results in either success or failure and the probability of success is the constant $q_{i}=1/\lceil \phi\cdot m_{i} \rceil$; the mean and variance of $x_{t}^{(i)}$ are $(1-q_{i})/q_{i}$ and $(1-q_{i})/q_{i}^{2}$. In our simulations, we consider three types of spot market that are respectively {\em fully}, {\em moderately}, and {\em poorly} saturated with bids; correspondingly, the parameter $\phi$ is set to 1, $\frac{1}{2.5}$, and $\frac{1}{5}$ and its value determines the number of spot jobs that newly arrive at $t$. Once the bid of a user is accepted, the assigned instance will dedicate $b=6$ slots (i.e., 30 minutes) to it. Among the users whose bids are accepted at $t=6\cdot h+i$, each will continue bidding at $t+6$ at a probability $1-\varrho$ and stop bidding at a probability $\varrho$; the value of $\varrho$ is chosen from a uniform distribution over $\{0.1, 0.3, 0.5\}$. The mean of $\varrho$ is 0.3; on average, at the beginning of each $t+6$, $30\%$ of the users accepted at $t$ will stop bidding while the remaining $70\%$ users will continue biding. We denote by $y_{t^{\prime}}^{(i)}$ the number of the spot users who are accepted at $t$ and will continue bidding at $t^{\prime}$ where $t^{\prime}=t+b=6\cdot (h+1)+i$. Thus, at the $i$-th group of servers, for all $t^{\prime}=6\cdot h+i$ where $h=1, 2, \cdots$, the total number of bids available is $x_{t^{\prime}}^{(i)}+y_{t^{\prime}}^{(i)}$; specially, for the initial slot $i$, there are only bids that newly arrive and the total number of bids is $x_{i}^{(i)}$.


\subsection{Idleness in On-demand Market}
\label{sec.on-demand-idleness}

In this subsection, we show the resource utilization when the CSP only provides on-demand instances.
This helps better perceive the advantage of selling the idle states of on-demand instances as spot instances, although such practice has been adopted by Amazon EC2. In particular, we will provide both experimental and theoretical results available.

\subsubsection{Empirical Results}
\label{sec.on-demand-instances-idle-empirical}

We implemented a queuing system to reproduce the utilization of servers.
The way of generating and dispatching on-demand jobs and the guaranteed quality of services (QoS) are described in Sec.~\ref{sec.job-arrivals} and Sec.~\ref{sec.components}.

We first look at the idleness of on-demand market when the expected job arrival rate $\lambda_{o}$, the shape parameter $\alpha$, and the upper bound of job size $\overline{x}$ take different values. In a pareto distribution, the larger the value of $\alpha$, the larger the expected job size; the latter is $\frac{\alpha\cdot x_{m}}{\alpha-1}$ for $\alpha>1$. The experiments are run over a period of about 120000 slots. While a server serves on-demand jobs, some super-slots are unoccupied; the idleness in a period is defined as the ratio of the amount of the unoccupied super-slots of all servers to the amount of the super-slots of all servers, denoted by $\vartheta$.
In the first case, we fix $\overline{x}=156$ and $\alpha=\frac{7}{6}$; the idleness $\vartheta$ is 0.8854, 0.8854, and 0.8856 when $\lambda_{o}$ equals 30, 60, and 90 respectively, which coincides with the theoretical result that the server utilization is independent of the job arrival rate given the QoS requirement \cite{Mitzenmacher01a}. In the second case, we fix $\lambda_{o}=60$ and $\overline{x}=156$, and the idleness $\vartheta$ is 0.8865, 0.8854, and 0.8769 when $\alpha$ equals $\frac{21}{20}$, $\frac{7}{6}$, and $\frac{7}{3}$. In the third case, we fix $\lambda_{o}=60$ and $\alpha=\frac{7}{6}$, and the idleness $\vartheta$ is 0.8854, 0.8877, and 0.8893 when $\overline{x}$ equals 156, 468, and 1404. 
The results of the latter two cases imply that the larger the job sizes, the higher the idleness of on-demand market. Under the different conditions above, we can observe that $\vartheta$ varies in a very small range of $[0.8769, 0.8893]$; the values of $\lambda_{o}$, $\alpha$, and $\overline{x}$ have slight effect on the idleness. Complementarily, the experimental results here are also consistent with a measurement study in \cite{Liu11a} where some instances of Amazon EC2 are launched and run for one week and the observed server utilization is in the 3\% to 17\% range.

For all $i\in [1, 6]$, recall the definition of $\theta_{t}^{(i)}$ in Section~\ref{sec.performance-metrics}. Let $\hat{\theta}_{t}^{(i)}=1-\theta_{t}^{(i)}$ and it denotes the percentage of servers idle and wasted at a super-slot of the $i$-th processing unit if the CSP does not sell them as spot instances. The average vacancy rate $\hat{\theta}^{(i)}$ of the $i$-th unit is defined as the average value of all $\hat{\theta}_{t}^{(i)}$ where $t=6\cdot h+i$ and $h=0, 1, 2, \cdots$. Now, we fix $\lambda_{o}=60$, $\alpha=\frac{7}{6}$, and $\overline{x}=156$ where the idleness is moderate; the values of $\hat{\theta}^{(1)}, \hat{\theta}^{(2)}, \cdots, \hat{\theta}^{(6)}$ are listed in Table~\ref{vacancy_rate}. For example, at the first unit, the percentage of servers idle at every super-slot is 88.56\% on average.

\begin{table}[t]
	\centering
		\caption{The vacancy rate of the six processing units when the CSP only offers on-demand instances}
		\label{vacancy_rate}
	\begin{threeparttable}[b]
		\begin{tabular}{|C{1.0cm}|C{1.0cm}|C{1.0cm}|C{1.0cm}|C{1.0cm}|C{1.0cm}|}	
			\hline
      $\hat{\theta}^{(1)}$   &  $\hat{\theta}^{(2)}$  &  $\hat{\theta}^{(3)}$  &   $\hat{\theta}^{(4)}$   &  $\hat{\theta}^{(5)}$  &  $\hat{\theta}^{(6)}$   \\ \hline	
      0.8856      &    0.8850     &    0.8857   &   0.8850  &  0.8853   &   0.8847   \\ \hline	
		\end{tabular}
	\end{threeparttable}
\end{table}


In the extended framework of Sec. 5, in order to have Feature~\ref{feature-21} and Feature~\ref{feature-22}, we divide servers into multiple groups that alternatingly serve the on-demand jobs arriving at different slots; however, in the basic framework of Sec. 3 and Sec. 4, servers are not divided: whenever an on-demand job arrives at any slot, one server will be chosen from all servers to serve it. Now, we show the effect of server division on the server utilization. Each experiment has the same job/workload input, guarantees the same QoS described in Sec.~\ref{sec.components}, and is taken respectively with and without server division. We will see that more servers are needed in the non-division scenario; thus, after division, the server utilization is improved and under the extended framework the servers are more effectively utilized by on-demand jobs. In particular, for the first case above, the minimum number of servers needed in the non-division scenario (resp. in the division scenario) is 6823, 13594, and 20421 (resp. 6105, 12198, and 18284); the corresponding utilization is improved by 11.76\%, 11.44\%, and 11.69\%. For the second case, the minimum number of servers needed in the non-division scenario (resp. in the division scenario) is 14589, 13594, and 8876 (resp. 13156, 12198, and 7450); the corresponding utilization is improved by 10.89\%, 11.44\%, and 19.14\%. For the third case, the minimum number of servers needed in the non-division scenario (resp. in the division scenario) is 13594, 15970, and 17837 (resp. 12198, 14452, and 16450); the corresponding utilization is improved by 11.44\%, 10.50\%, and 8.432\%.




\subsubsection{Theoretical Results}\label{sec.resource-management}

Results from discrete-time queuing theory can be used to help us perceive the relation between the mean waiting time of on-demand jobs and the utilization of servers; The standard definition for a job's waiting time is the queuing time from its arrival to the moment that it gets assigned. In particular, existing literature considers the case where the job arrival at a server follows a geometric distribution and the job size follows a general distribution. In the context of this paper, we can use the round-robin policy in Section~\ref{sec.on-demand-jobs} to uniformly dispatch the arriving on-demand jobs to servers; then, at every server there is a single queue and the mean waiting time of all on-demand jobs will be its counterpart at a server \cite{Zheng16a}.
We denote by $\sigma$ the job size's standard deviation and by $s$ the mean job size. At a server, the mean waiting time $w$ satisfies the following relation \cite{discrete-time-queue}:
\begin{align}\label{discrete-result}
w = \frac{\lambda\cdot (\sigma^{2}+s^{2}) - \rho}{2\cdot (1-\rho)}
\end{align}
where one job arrives at a slot with probability $\lambda$ and the probability that no jobs arrive is $1-\lambda$ where $\lambda \in [0,1]$; $\rho$ is the mean utilization or load of a server where $\rho=\lambda\cdot s$. By (\ref{discrete-result}), we also have
\begin{align}\label{discrete-result-1}
\frac{1}{\rho} = 1 + \frac{\sigma^{2}/s + s - 1}{2\cdot w}.
\end{align}
In the following, we illustrate the sensitivity of the server utilization $\rho$ to the waiting time $w$: the requirement of a small waiting time leads to a low utilization in the on-demand market. In this section, all servers are divided 6 groups and each group can be viewed as a single system where jobs arrive and are dispatched once every super-slot (i.e., 6 slots). To apply the relation (\ref{discrete-result-1}) directly, we use in this subsubsection the super-slot as the basic time unit for the job size. Since on-demand jobs are charged on an hourly basis, their size will be the multiple of 2 (super-slots); thus $s\geq 2$ and we have by (\ref{discrete-result-1}) that $\rho$ decreases as $w$ decreases. We assume that, the CSP will guarantee that the mean waiting time $w$ is $\frac{1}{6}$ super-slot. When the job size follows a uniform distribution over $\{2, 4, 6\}$, we have that the mean job size $w$ is $4$ and its variance $\sigma^{2}$ is small and equals 2; then, to guarantee the QoS, we have by (\ref{discrete-result-1}) that the server utilization is $\frac{1}{11.5}\approx 0.08696$. In other words, on average, the servers will be in idle state 91.30\% of the time. Thus, many servers are in idle states in the on-demand market, which remain to be utilized by spot users.

\subsection{Revenue Improvement of Spot Market}
\label{sec.basic-spot-pricing}

\begin{table}[t]
	\centering
		\caption{The revenue improvement when the bid prices follow a uniform distribution over $[0.2, 1]$}
		\label{efficiency_improvement-1-0}
	\begin{threeparttable}[b]
		\begin{tabular}{|C{1.5cm}|C{1.5cm}|C{1.5cm}|}	
			\hline
			            $\alpha_{e}^{(1)}$   &  $\alpha_{e}^{(2.5)}$  &  $\alpha_{e}^{(5)}$     \\ \hline			
			              4.615      &    2.894     &    1.641      \\ \hline	
		\end{tabular}
	\end{threeparttable}
\end{table}

\begin{table}[t]
	\centering
		\caption{The revenue improvement when the bid prices follow a bounded pareto distribution in $[0.3, 1]$ with a pareto index of 2.}
		\label{efficiency_improvement-2-0}	
    \begin{threeparttable}[b]
		\begin{tabular}{|C{1.5cm}|C{1.5cm}|C{1.5cm}|}	
			\hline
			            $\alpha_{e}^{(1)}$   &  $\alpha_{e}^{(2.5)}$  &  $\alpha_{e}^{(5)}$     \\ \hline			
			              3.036      &    1.990     &    1.638      \\ \hline	
		\end{tabular}
	\end{threeparttable}
\end{table}

In this subsection, we show the main results of performance evaluation, i.e., the revenue improvement $\alpha_{e}$ as explained in Sec.~\ref{sec.performance-metrics}. We use $\alpha_{e}^{(1)}$, $\alpha_{e}^{(2.5)}$, $\alpha_{e}^{(5)}$ to denote the $\alpha_{e}$ respectively in the case that the spot market is fully, moderately and poorly saturated with bids as explained in Sec.~\ref{sec.job-arrivals-1}. The main experimental results are given in Table~\ref{efficiency_improvement-1-0} and~\ref{efficiency_improvement-2-0}. For example, in the case that the users' bid prices follow a uniform distribution, if the spot market is saturated with many bids (i.e., the fully-saturated case), there is a at least 4.5-fold increase in the CSP's revenue; if the spot market is less saturated (i.e., the poorly-saturated case), the CSP's revenue can still be increased by more than 1.5-fold.

For the uniform distribution case, the average value of all spot prices is 0.7756, 0.6624, and 0.6061 respectively in the fully, moderately, and poorly saturated spot market. The spot price is the minimum price of all accepted bids. With more bids available, the CSP can choose to accept the bids with higher prices; hence, in a more saturated spot market, the average spot price is also higher. The spot prices of the first unit at slot $t=222001+6\cdot h^{\prime}\in [222001, 222300]$ is illustrated in Fig.~\ref{fig.spot_prices_over_time} where $h^{\prime}=0, 1, \cdots, 49$; we can observe that, if the spot market is saturated to a higher degree, the change of spot prices over time is also larger. Taking all the 6 processing units into account, the average number of spot jobs accepted per super-slot is 1444, 1065, and 666 respectively in the fully, moderately, and poorly saturated spot market; correspondingly, the average number of spot jobs newly arriving and accepted is 506.1, 335.1, and 201.2. Let us consider the poorly and fully saturated markets respectively. Due to the offer of spot instances to users, in the former case, the average utilization of instances is improved from 0.1146 to 0.4441, while the CSP's revenue is improved by 164.1\%; in the latter case, the average utilization of instances is improved from 0.1146 to 0.8292, while the CSP's revenue is improved by 461.5\%. As far as the first unit is concerned, the number of spot instances available and the number of accepted bids at slot $t=222001+6\cdot h^{\prime}\in [222001, 222300]$ are illustrated in Fig.~\ref{number_preempt_basic} where $h^{\prime}=0, 1, \cdots, 49$.

\begin{figure}[t]
\begin{center}
  \includegraphics[width=3.15in]{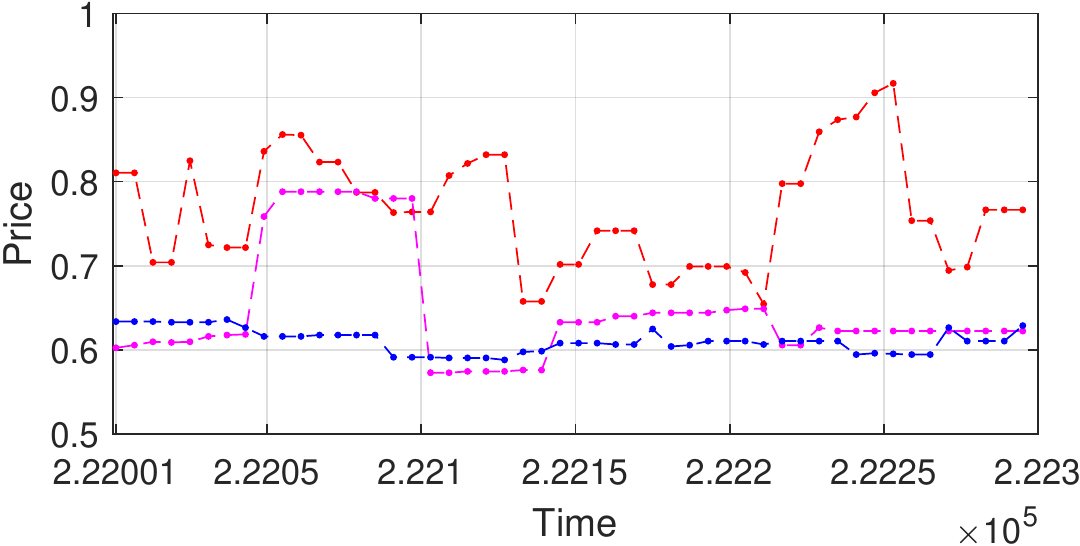}
  \caption{The spot prices over $[222001, 222300]$: the red, magenta and blue points correspond to the fully, moderately and poorly saturated case respectively.}\label{fig.spot_prices_over_time}
\end{center}
\end{figure}

\begin{figure}[t]
\begin{center}
  \includegraphics[width=3.35in]{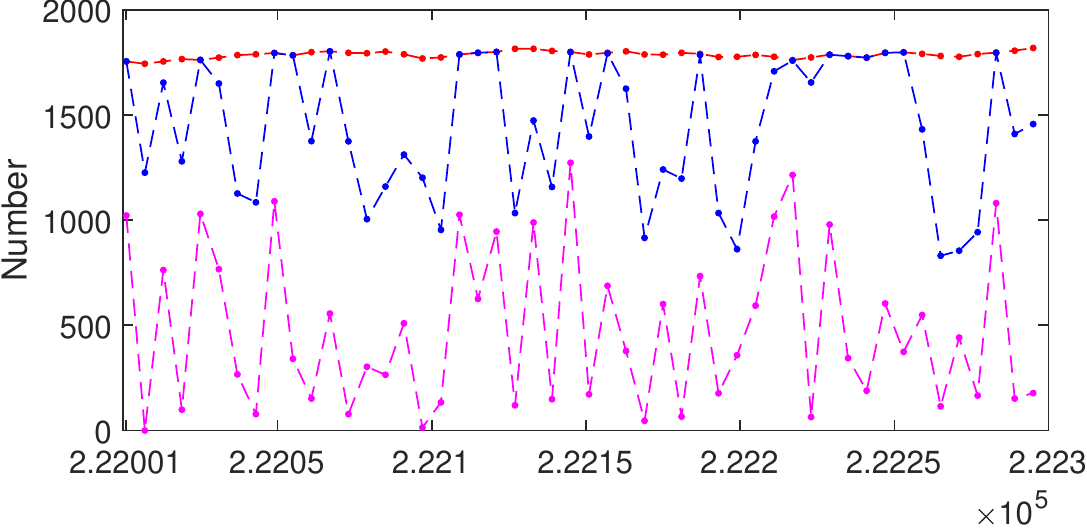}\\
\includegraphics[width=3.35in]{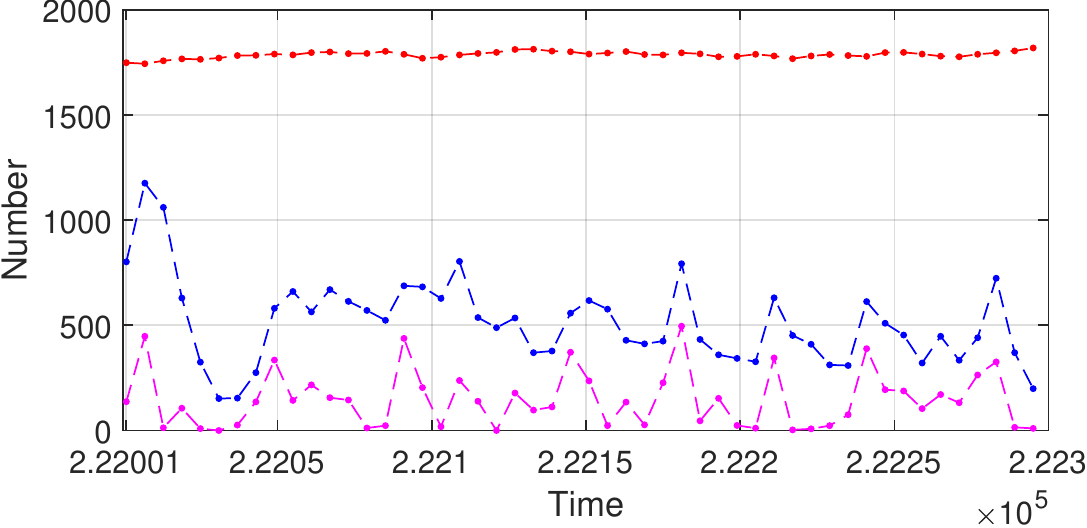}
  \caption{The supply and demand relation in the slot interval $[222001, 222300]$: the red points denote the number of spot instances, the blue points denote the total number of bids accepted, and the magenta points denote the number of bids newly arriving and accepted; the top and bottom subfigures correspond to the fully and poorly saturated cases respectively.}\label{number_preempt_basic}
\end{center}
\end{figure}

Finally, we summarize the average server utilization per super-slot before and after offering spot instances to users. Recall the definition of $\theta$ and $\mu$ in Sec.~\ref{sec.performance-metrics}; we use $\mu_{x}^{(y)}$ denote the value of $\mu$ under a specific environment. In particular, when $x=1$ ({\em resp.} $x=2$), the experiment is done in the case that the bid prices follow the uniform distribution ({\em resp.} the bounded pareto distribution). $\mu_{x}^{(1)}$, $\mu_{x}^{(2.5)}$, and $\mu_{x}^{(5)}$ denote the utilization respectively in the case that the spot market is fully, moderately, and poorly saturated. For example, $\mu_{1}^{(1)}$ denotes the utilization when the bid prices follow the uniform distribution and the spot market is fully saturated. The related results are listed in Table~\ref{utilization_improvement}.

\begin{table*}[t]
	\centering
		\caption{The improvement to resource utilization: in case 1, the CSP only offers on-demand instances; in cases 2 and 3, both on-demand and spot instances are offered, where the bid prices follow a uniform distribution and a bounded pareto distribution respectively.}
		\label{utilization_improvement}
	\begin{threeparttable}[b]
		\begin{tabular}{|C{1.5cm}|C{1.5cm}|C{1.5cm}|C{1.5cm}|C{1.5cm}|C{1.5cm}|C{1.5cm}|}	
			\hline
\cline{1-7} \multicolumn{1}{ |c| }{Case 1} &  \multicolumn{3}{ c| }{Case 2} & \multicolumn{3}{ c| }{Case 3}\\ \hline
			         $\theta$   &   $\mu_{1}^{(5)}$   &  $\mu_{1}^{(2.5)}$  &  $\mu_{1}^{(1)}$  &  $\mu_{2}^{(5)}$   &  $\mu_{2}^{(2.5)}$  &  $\mu_{2}^{(1)}$   \\ \hline			
			    0.1146      &    0.4441      &   0.6415     &    0.8292    &   0.4435      &    0.7411     &    0.8309   \\ \hline	
		\end{tabular}
	\end{threeparttable}
\end{table*}

\subsection{Comparison with a Dynamic Reserve Price Algorithm}

The framework for sharing the time of servers among on-demand and spot services -- in order to optimally pricing spot instances -- has four main components as described in Section~\ref{sec.components}. The $4$-th component is actually the one operating the spot pricing mechanism: consistent with the claim of Amazon EC2 \cite{Ben-Yehuda13a}, it should rely on the relation of demand and supply, i.e., the bids and the number of spot instances available. We aim at comparing our scheme with the reference one proposed by the authors of \cite{Ben-Yehuda13a}; they actually claim that Amazon EC2 may in practice set its spot prices artificially by a dynamic reserve price (DRP) algorithm \cite{Ben-Yehuda13a}. In this subsection, we thus replace the pricing scheme in the 4th component with the DRP algorithm in \cite{Ben-Yehuda13a}, and show the performance of our framework when the scheme in \cite{Ben-Yehuda13a} is applied. The DRP algorithm draws spot prices from a fixed range $[F, C]$ where $F$ and $C$ are the lower and upper bounds of spot prices. In particular, it is initialized with a reserve price of $P_{0} = F$ and a price change of $\Delta_{0} = 0.1\cdot (F-C)$. At each processing unit, the spot price is updated every 6 slots and the $l$-th spot price $P_{l}$ is recursively defined as follows where $l=1,2,\cdots$:
\begin{center}
$P_{l} = P_{l-1} + \Delta_{l}$, and $\Delta_{l} = -0.7 \cdot \Delta_{l-1} +\epsilon(\sigma^{\prime})$
\end{center}
where $\epsilon(\sigma^{\prime})$ is white noise with a standard deviation $\sigma^{\prime}=0.39\cdot(C-F)$; here, $\Delta_{l}$ may be generated multiple times until the resulting $P_{l}$ is within $[F, C]$ and does not equal $P_{l-1}$, i.e., $P_{l}\in [F, C]-\{P_{l-1}\}$. In our experiments of Section~\ref{sec.basic-spot-pricing}, if the users' bid prices follow the uniform distribution, the minimum and maximum spot prices are $(\pi_{min}, \pi_{max})=$ $(0.5024, 0.9496)$, $(0.5384, 0.8744)$ and $(0.5224, 0.7712)$ respectively in the fully, moderately, and poorly saturated case; correspondingly, if the bid prices follow the Pareto distribution, $(\pi_{min}, \pi_{max})$ is $(0.3000, 0.8110)$, $(0.3000, 0.6110)$ and $(0.5048, 0.7272)$ respectively. In the experiments, we set $(F, C)=(\pi_{min}, \pi_{max})$ in each case.
The related results for revenue improvement are listed in Table~\ref{efficiency_improvement-1-1} and~\ref{efficiency_improvement-2-1}, and the server utilization before and after complementing the on-demand market with a spot market is summarized in Table~\ref{utilization_improvement-2}; the meaning of related symbols has been introduced in Section~\ref{sec.basic-spot-pricing}.

\begin{table*}[t]
	\centering
		\caption{With the pricing scheme in \cite{Ben-Yehuda13a}, the revenue improvement when the bid prices follow the uniform distribution.}
		\label{efficiency_improvement-1-1}
	\begin{threeparttable}[b]
		\begin{tabular}{|C{1.5cm}|C{1.5cm}|C{1.5cm}|}	
			\hline
			            $\alpha_{e}^{(1)}$   &  $\alpha_{e}^{(2.5)}$  &  $\alpha_{e}^{(5)}$     \\ \hline			
			              3.526      &    2.182     &    1.373      \\ \hline	
		\end{tabular}
	\end{threeparttable}
\end{table*}

\begin{table*}[t]
	\centering
		\caption{With the pricing scheme in \cite{Ben-Yehuda13a}, the revenue improvement when the bid prices follow the bounded pareto distribution.}
		\label{efficiency_improvement-2-1}
	\begin{threeparttable}[b]
		\begin{tabular}{|C{1.5cm}|C{1.5cm}|C{1.5cm}|}	
			\hline
			            $\alpha_{e}^{(1)}$   &  $\alpha_{e}^{(2.5)}$  &  $\alpha_{e}^{(5)}$     \\ \hline			
			              1.960      &    1.375     &    0.4620      \\ \hline	
		\end{tabular}
	\end{threeparttable}
\end{table*}

\begin{table*}[t]
	\centering
		\caption{The improvement to resource utilization when the spot pricing scheme in \cite{Ben-Yehuda13a} is applied: in case 1, the CSP only offers on-demand instances; in cases 2 and 3, both on-demand and spot instances are offered, where the bid prices follows a uniform distribution and a bounded pareto distribution respectively.}
		\label{utilization_improvement-2}
	\begin{threeparttable}[b]
		\begin{tabular}{|C{1.5cm}|C{1.5cm}|C{1.5cm}|C{1.5cm}|C{1.5cm}|C{1.5cm}|C{1.5cm}|}	
			\hline
\cline{1-7} \multicolumn{1}{ |c| }{Case 1} &  \multicolumn{3}{ c| }{Case 2} & \multicolumn{3}{ c| }{Case 3}\\ \hline
			         $\theta$   &   $\mu_{1}^{(5)}$   &  $\mu_{1}^{(2.5)}$  &  $\mu_{1}^{(1)}$  &  $\mu_{2}^{(5)}$   &  $\mu_{2}^{(2.5)}$  &  $\mu_{2}^{(1)}$   \\ \hline			
			    0.1146      &    0.3769      &   0.5004     &    0.7255    &   0.2084      &    0.5164     &    0.5968   \\ \hline	
		\end{tabular}
	\end{threeparttable}
\end{table*}


Overall, we can see from Table~\ref{efficiency_improvement-1-0} and \ref{efficiency_improvement-1-1} (or Table~\ref{efficiency_improvement-2-0} and \ref{efficiency_improvement-2-1}) that, the proposed algorithm of this paper can achieve higher revenue improvement than the DRP algorithm in \cite{Ben-Yehuda13a}; here, the revenue from the on-demand market depends on the utilization of on-demand market that equals 0.1146, as shown in Table \ref{utilization_improvement} and \ref{utilization_improvement-2}. As claimed in \cite{Ben-Yehuda13a}, when the spot market is saturated with less bids, the use of the DRP scheme can artificially create a false impression of the changes of demand and supply and mask times of low demand and price inactivity, thus possibly driving up the CSP's stock. This is confirmed in our experimental results: as illustrated in Fig.~\ref{fig.spot_prices_over_time}, when the spot market is saturated with many bids (e.g., the fully-saturated case), the spot prices vary more dramatically over time; however, in the poorly-saturated case, the spot prices vary slightly and even keep constant in a relatively long period. So, in the poorly-saturated case, it may be necessary to sets spot price artificially.

\subsection{Comparison with Another Framework for On-demand and Spot Services}

We first introduce the framework of Dierks {\em et al.} in \cite{Dierks16a,Dierks19a} where on-demand and spot markets are modeled as two separate queues $Q_{1}$ and $Q_{2}$. For $Q_{1}$, the number of servers $m_{o}$ is the minimum severs needed to guarantee that the expected waiting time of jobs is very small; the price of utilizing a sever is $p$ per unit of time. The second $Q_{2}$ is a priority queue that has $m_{S}$ servers: every job continuously bids to utilize the servers of spot market until it gets enough execution time; jobs with higher bid prices have higher priorities to utilize servers. There are $n$ job classes whose expected arrival rates are $\lambda=(\lambda_{1}, \cdots, \lambda_{n})$. The job size is drawn from a probability distribution with expectation $\frac{1}{\mu}$. For every job of class $i$, a waiting cost $c$ is drawn from a distribution $F_{i}(c)$ on $[0,\, v_{i}^{\prime}]$.
For a job of spot market, its execution on a server may be preempted when there are unfinished jobs of higher bid prices; each preemption brings a cost $c\cdot\tau$ to it.
Let $\sigma=(\zeta, \eta)$ and $\sigma$ denotes the strategy of a job: when $\zeta=\mathcal{O},\, \mathcal{S}$, or $\mathcal{B}$, it means that this job will choose the on-demand, spot, or neither market; $\eta$ is the bid price if $\zeta=\mathcal{S}$.
In a Bayesian Nash incentive compatible spot market framework, all participants ideally have the knowledge such as $\lambda$, $\mu$ and $F_{1}(c), \cdots, F_{n}(c)$. Further, each job of class $i$ can derive the optimal strategy $\sigma$ to maximize its expected payoff; here, if a job chooses spot market, its bid price will be its waiting cost $c$ and its payment can also be derived. With the payments of jobs, we can get the revenue of both markets.

The framework of this paper elaborates on the current service model in Amazon EC2. Compared with \cite{Dierks19a}, the Amazon EC2 model has its advantages in terms of revenue generation and quality of service. It allows selling the idle state of on-demand market on the spot market to get additional revenue \cite{Devanur17,Ben-Yehuda13a}.
A complication of the model in \cite{Dierks19a} lies in that the completion time of a spot job depends on the arrival rate of the jobs of higher bid prices that is usually uncertain in reality.
This incurs the volatility of jobs' completion times and even delay-tolerant users can become reluctant to accept such service. As discussed in Sec.~\ref{sec.preliminary}, using the model of this paper, delay-tolerant jobs can first bid to utilize spot instances in some period and then turn to stable on-demand instances, leading to that they are finished at expectable times. The numerical comparison of revenue should be taken under the same input condition. In \cite{Dierks19a}, once a spot job submits the bid, it cannot cancel its bid until it gets a specific amount of execution time; its payment relies on its completion time and waiting cost. In our framework, a spot user can stop bidding at any latter slot after it begins bidding; the payment at a slot is the minimum price of all accepted bids, with no connection to the job's waiting cost and completion time. The CSP's revenue is the price times the processed workload. To enable comparison, we let the average spot prices in both frameworks be the same. Specifically, we let the on-demand price be $1$; the generated revenue is 0.5 when a spot instance dedicates an hour to a job. As before, we use the PTC policy to assign jobs to servers.


\begin{table*}[t]
	\centering
		\caption{Revenue improvement compared with the on-demand and spot model in \cite{Dierks19a}.}
		\label{efficiency_improvement-priority-queue}
	\begin{threeparttable}[b]
		\begin{tabular}{|C{1.5cm}|C{1.5cm}|C{1.5cm}|}	
			\hline
			            $\alpha_{e}^{(30)}$   &  $\alpha_{e}^{(60)}$  &  $\alpha_{e}^{(120)}$     \\ \hline			
			              2.816      &    2.181     &    1.736      \\ \hline	
		\end{tabular}
	\end{threeparttable}
\end{table*}

In \cite{Dierks19a}, we set $n=1$ and the waiting costs of jobs are uniformly distributed on $\{0.3, 0.7\}$. The job arrivals follow a poisson distribution, with expectations $\lambda_{o}^{\prime}$ in on-demand market and $\lambda_{s}^{\prime}$ in spot market. Fixing $\lambda_{o}^{\prime}=30$, we consider three cases with $\lambda_{s}^{\prime}=$ 30, 60, and 120 respectively. The sizes of jobs are set as described in Sec.~\ref{sec.job-arrivals}. Jobs of spot market have to be finished within bounded periods. We let jobs with $c=0.7$ have a waiting time $\leq 48$ slots (i.e., 4 hours) and let jobs with $c=0.3$ have a waiting time $\leq 132$ slots (i.e., 11 hours); the value of $m_{S}$ is the minimum amount of servers needed to guarantee the QoS, i.e., the spot market is fully saturated with as many jobs as possible; the QoS is guaranteed at the 99th percentile as described in Sec.~\ref{sec.components}. The total number of servers is $m=m_{o}+m_{s}$. We compute the average numbers of servers utilized per slot in on-demand and spot markets, denoted by $\overline{m}_{o}$ and $\overline{m}_{s}$; the average unit revenue of both markets is $\left(\overline{m}_{o}+0.5\cdot\overline{m}_{s}\right)/12$. For comparison, we also compute the average unit revenue achieved under our framework in the case that there are $m$ servers; here, the spot market is also fully saturated with all idle servers of on-demand market being utilized by spot users. We denote by $\hat{\alpha}_{e}$ the ratio of the average unit revenue of our framework to its counterpart with the framework of \cite{Dierks19a}; let $\alpha_{e}=\hat{\alpha}_{e}-1$ denote the revenue improvement when comparing our framework with the one in \cite{Dierks19a}. We use $\alpha_{e}^{(30)}$, $\alpha_{e}^{(60)}$, $\alpha_{e}^{(120)}$ to denote the $\alpha_{e}$ respectively in the case that $\lambda_{s}^{\prime}$ is 30, 60, and 120, and the experimental results are given in Table~\ref{efficiency_improvement-priority-queue}.



\section{Conclusion}
\label{sec.conclusion}

In a system where on-demand and spot users coexist, on-demand users arrive randomly and have high priority to access servers, while spot users bid prices to utilize the time periods unoccupied by on-demand users. A key feature to make such services accessible is that, on-demand users can get served within a short time upon arrivals, while a spot user can stably utilize a server (without the interference of on-demand users) for a sufficient amount of time once bidding successfully. In this paper, we propose a framework that has such a feature for sharing the time of servers among on-demand and spot users. Under such a framework, specific schemes are proposed to accept and assign the requests of on-demand and spot users to servers and to optimally price spot instances. 
The framework itself is designed under assumptions which are met in real environments. With a few further mild assumptions, an analysis of the proposed framework is also taken to understand which parameters drive its performance. Extensive simulations show a significant improvement to the revenue as well as the server utilization once an on-demand market is complemented with a spot market. In the case where less bids are available, the revenue improvement is showed to be indeed smaller, but still significant compared with a bare on-demand market.

\bibliographystyle{ACM-Reference-Format}
\bibliography{sample-base}

\end{document}